\documentclass{article} 
\usepackage{iclr2027_conference,times}

\usepackage{amsmath,amsfonts,bm}

\def\eqref#1{equation~\ref{#1}}

\def\1{\bm{1}}

\DeclareMathAlphabet{\mathsfit}{\encodingdefault}{\sfdefault}{m}{sl}
\SetMathAlphabet{\mathsfit}{bold}{\encodingdefault}{\sfdefault}{bx}{n}

\usepackage{hyperref}
\usepackage{url}
\usepackage{booktabs}
\usepackage{multirow}
\usepackage{tablefootnote}
\usepackage{wrapfig}
\usepackage[table]{xcolor}
\usepackage{etoolbox}

\usepackage{hyperref}
\usepackage{url}
\usepackage{mathtools}

\usepackage{amsmath,amssymb,amsthm}
\usepackage{algorithm,algpseudocode}
\usepackage{enumitem}
\usepackage{multirow} 

\usepackage{booktabs}
\usepackage{graphicx}
\usepackage{diagbox}
\usepackage{caption}
\usepackage{subcaption}

\usepackage{tcolorbox}

\usepackage{xparse}
\usepackage{etoolbox}

\usepackage{wrapfig}
\usepackage[normalem]{ulem}

\newtheorem{proposition}{Proposition}[section]

\newif\ifhighlight
\highlighttrue

\newif\ifarxiv
\arxivtrue      

\ifarxiv
    \iclrfinalcopy

    \makeatletter
    \patchcmd{\@maketitle}
      {\lhead{Published as a conference paper at ICLR 2027}}
      {\lhead{Preprint}}
      {}{}
    \makeatother
\fi

\title{Language Model Fingerprinting Requires \\Rethinking Watermark Teachers}

\author{
Jeongyeon Hwang$^{1}$,
Anshul Nasery$^{2}$,
Sewoong Oh$^{2}$,
Jungseul Ok$^{1}$\thanks{
Correspondence to: Jungseul Ok
\textless\texttt{jungseul@postech.ac.kr}\textgreater.
} \\
$^{1}$Pohang University of Science and Technology
\qquad
$^{2}$University of Washington
}

\begin{document}

\maketitle



\begin{abstract}
LLM fingerprinting via watermark distillation embeds a statistical watermark signal into model weights, enabling model owners to identify their models behind black-box APIs. Revisiting a recent protocol, we find that its utility evaluation understates text quality degradation in open-ended generation, favoring overly strong watermark teachers. Weakening the watermark improves text quality but sacrifices detectability. To move beyond this trade-off, we rethink whether text watermarks designed for verifying generated text are suitable distillation teachers for model fingerprinting. Such watermarks are typically designed to remain detectable from an individual output, limiting how sparse the watermark signal can be. In contrast, fingerprint verification can aggregate signal across queries, making sparser watermark signals viable. This raises a key question: where should the sparse signal be placed? We analyze signal placement through token surprisal and show that, even at comparable watermark strength, different placements can target tokens with different plausibility under the base model. This motivates near-tie restriction, which uses top-1-relative logit gaps to restrict the watermark bias to tokens close to the base model's top prediction. Across multiple models, near-tie improves detection--quality frontiers under deployment changes, preserves higher text quality across query budgets, and further improves existing watermarking schemes when combined with them.
\end{abstract}

\vspace{-1em}
\section{Introduction}\label{sec:introduction}
\vspace{-.5em}
The strong performance and widespread adoption of large language models (LLMs) have intensified the competition to develop increasingly capable models, with a growing number of such models released with open weights~\citep{minaee2024large,qwen2025qwen25}. Once released, however, these models may be deployed behind black-box APIs without honoring restrictive licenses, raising concerns about model ownership and unauthorized use~\citep{white2024model,xu2024instructional}. 
To build a healthy ecosystem, there is a growing need for reliable model provenance verification.

Recently, black-box model fingerprinting methods~\citep{xu2024instructional,russinovich2024heythatsmodelintroducing,nasery2026scalable,gloaguen2026llm} have been proposed for this purpose. They broadly fall into backdoor-based~\citep{xu2024instructional,russinovich2024heythatsmodelintroducing,nasery2026scalable} and watermark-based~\citep{gloaguen2026llm} approaches. In this work, we focus on watermark-based fingerprinting, which uses a text watermark originally designed to verify the provenance of generated text~\citep{kirchenbauer2023watermark} as a distillation teacher. Its watermarked output distribution is distilled into model weights on a designated fingerprint domain, such as French, so that queries from this domain elicit a detectable statistical signal. Compared with backdoor-based approaches, this framework offers three advantages: ($i$) statistical detection with false-positive control, ($ii$) natural-looking outputs that are challenging to detect and filter out, and ($iii$) robustness to deployment transformations such as decoding variations, quantization, pruning, and fine-tuning.

Despite these advantages, we find that the utility evaluation of a recent protocol substantially understates degradation in open-ended generation within the fingerprint domain.
This favors overly strong watermark teachers that degrade generation quality for legitimate users in that domain. Reducing watermark strength improves text quality but sacrifices detectability, while more advanced watermarking schemes provide only limited gains when used as distillation teachers (Figure~\ref{fig:main_frontier}).

To move beyond this trade-off, we ask whether text watermarks are well suited as distillation teachers for model fingerprinting. Text watermarks~\citep{kirchenbauer2023watermark,lee2024wrote,wang2025morphmark,wu2023dipmark,liu2024unforgeable,kuditipudi2023robust} are typically designed to remain detectable from an individual output, including short or modified outputs, which constrains how sparse their signal can be. Model fingerprinting, by contrast, can aggregate signal across multiple queries, making sparser teacher signals viable. This raises a key question: where should the remaining signal be placed? We analyze signal placement through token surprisal and show that, even at comparable watermark strength, different placements target tokens of different plausibility. Motivated by this analysis, we propose near-tie restriction, which uses logit gaps to limit teacher modifications to tokens close to the model's top prediction. Since it only restricts which tokens carry the watermark signal, near-tie can be incorporated into existing watermarking schemes without altering their underlying mechanisms.

Empirically, this simple restriction substantially improves detection--quality frontiers across multiple models and watermarking schemes under common deployment transformations. It also yields higher generation quality across query budgets, and further improves existing watermarking schemes when combined with them.

\begin{figure}[t]
    \centering
    \includegraphics[width=1.0\linewidth]{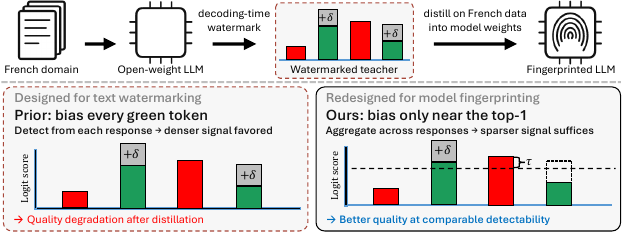}
    \vspace{-1.8em}
    \caption{Overview of rethinking watermark teachers for model fingerprinting. Top: watermark-based fingerprinting distills a text watermark into model weights on a semantic domain such as French. The watermark splits tokens into green and red sets and biases green tokens to create a detectable statistical signal. Bottom: prior teachers bias every green token for individual-response detection, whereas fingerprint verification aggregates evidence across responses and permits sparser signals. Our near-tie restriction applies the bias only to green tokens near the top-1 logit.}
    \label{fig:overview}
    \vspace{-1.4em}
\end{figure}

Our main contributions are summarized as follows:
\begin{itemize}
\item We show that the evaluation of a recent watermark-based fingerprinting protocol understates text quality degradation in open-ended generation and favors overly strong teachers.

\item We point out an overlooked implication of model fingerprinting: aggregating signal across queries makes sparser teacher signals viable. Motivated by token-surprisal analysis, we propose near-tie restriction using base-model logit gaps for signal placement.

\item Our comprehensive experiments show that near-tie restriction consistently improves detection--quality frontiers, preserves higher generation quality across query budgets, and further improves existing watermarking schemes when combined with them.

\end{itemize}

\vspace{-1em}
\section{Related Work}\label{sec:related_work}
\vspace{-.5em}
\textbf{LLM fingerprinting.}
White-box fingerprinting methods~\citep{refael2024slip,zeng2024huref,zhang2025reef} provide robust verification but require access to model weights, limiting their practicality in real-world deployments. This motivates black-box methods, which verify ownership from queries and responses alone and broadly follow either a backdoor-based or watermark-based approach. Backdoor-based methods fine-tune a model to emit predetermined responses to secret trigger queries and verify ownership through exact matches~\citep{xu2024instructional,nasery2026scalable,russinovich2024heythatsmodelintroducing}. However, atypical query--response pairs can undermine stealthiness, while exact-match verification can be fragile to deployment transformations. Watermark-based fingerprinting instead distills a watermark into model weights, enabling statistical verification from queries in a semantic domain and improving robustness to deployment transformations~\citep{gloaguen2026llm}. However, we find that its evaluation understates text quality degradation on open-ended generation within the fingerprint domain. We provide extended related work in Appendix~\ref{appendix:extended:related_work}.

\textbf{LLM watermarking.}
LLM watermarking supports text provenance by embedding a hidden statistical
signal into generated text. KGW~\citep{kirchenbauer2023watermark}, a widely
used red--green scheme, partitions the vocabulary and biases green-token
logits during decoding, making green tokens statistically overrepresented.
Subsequent work develops alternative partitioning and biasing strategies to
improve the detection--quality trade-off~\citep{lee2024wrote,wang2025morphmark,
wu2023dipmark,liu2024unforgeable}. Beyond decoding-time watermarking,
\citet{gu2024learnability} show that watermarks can be distilled into
open-weight models, enabling watermarked generation without modifying
inference-time decoding. Building on this idea, \citet{gloaguen2026llm}
distill a semantically conditioned watermark for model fingerprinting.
However, this directly inherits teachers designed for text provenance detection, overlooking that fingerprint verification aggregates evidence across responses. We exploit this distinction to design a sparser teacher that better preserves text quality.
\vspace{-1em}
\section{Preliminary}\label{sec:preliminary}
\vspace{-.5em}
A language model $\mathcal{M}$ generates text by predicting the next token $x_t$ in a sequence. Given the previous tokens $x_{<t}$, it outputs a vector of logits $\boldsymbol{\ell}_t$ over the vocabulary $\mathcal{V}$, which defines the next-token distribution $Q_t=\operatorname{softmax}(\boldsymbol{\ell}_t)$. The next token $x_t$ is then randomly sampled according to $Q_t$ depending on the choice of the decoding strategy. 

\textbf{Watermarking algorithm.}
LLM watermarking enables model owners to statistically verify whether a text was generated by their model or not. A common technique is to embed a detectable signal during decoding. In this work, we focus on  KGW introduced by \citet{kirchenbauer2023watermark}, which is a widely used red--green watermarking scheme. At each decoding step, a pseudorandom function of the secret key $\xi$ and the preceding $h$ tokens partitions the vocabulary $\mathcal{V}$ into a green list $G_t$, with $|G_t|=\gamma|\mathcal{V}|$, and its red complement. KGW adds a bias $\delta\geq 0$ to the logits of green tokens, yielding
\begin{equation*}
\widetilde{\ell}_t(v)
=
\ell_t(v)
+
\delta\mathbf{1}\{v\in G_t\}, \text{ for all $v\in \mathcal{V}$, and }
\qquad
\widetilde{Q}_t
=
\operatorname{softmax}(\widetilde{\boldsymbol{\ell}}_t)\;,
\end{equation*}
where we omit the dependence of $G_t$ on $[x_{t-h-1},\ldots,x_{t-1}]$ and $\xi$ for brevity. 
For a generated response $x=[x_0,\ldots,x_{L-1}]$, let $N_G(x)=\sum_{t=0}^{L-1}\mathbf{1}\{x_t\in G_t\}$ denote the number of generated green tokens. Under the null hypothesis of unwatermarked generation, the expected green-token count is $\gamma L$. The detector therefore computes a test statistic, 
\begin{equation}
z(x)
\; = \; 
\frac{N_G(x)-\gamma L}
{\sqrt{\gamma(1-\gamma)L}}\;,
\label{eq:watermark-z-score}
\end{equation}
and classifies the response as watermarked if $z(x)>\rho$, where $\rho$ denotes the detection threshold.

\textbf{Watermark-based fingerprinting.}
Watermark-based fingerprinting distills a decoding-time watermark into model weights. The distilled model emits the watermark signal without the need for modified decoding at inference time,  thereby enabling black-box ownership verification~\citep{gloaguen2026llm}. Unlike LLM watermarking for text provenance, which is required to watermark every output text, model fingerprinting only needs to watermark when prompted by the verifier. This is achieved by conditioning the fingerprint on a secret semantic domain $\mathcal{D}_{\mathrm{fp}}$, e.g., French, and a fingerprinted model $\mathcal{M}_{\mathrm{fp}}(\widetilde Q)$ is trained to match a watermarked teacher $\widetilde Q$ within this domain, minimizing 
\begin{equation}
\mathcal{L}_{\mathrm{fp}}(\theta)
\;=\;
\mathbb{E}_{x\sim\mathcal{D}_{\mathrm{fp}}}
\left[
\sum_{t=1}^{|x|}
\mathrm{KL}
\left(
\widetilde{Q}_t \Vert P_{\theta,t}
\right)
\right],
\label{eq:fingerprint-distillation}
\end{equation}
where $P_{\theta,t}=P_\theta(\cdot\mid x_{<t})$ denotes the distilled next-token distribution of $\mathcal{M}_{\mathrm{fp}}(\widetilde Q)$ parametrized by $\theta$ and $KL(\cdot\|\cdot)$ denotes the Kullback-Leibler divergence. A regularization term is additionally used to suppress watermarked generations outside $\mathcal{D}_{\mathrm{fp}}$ to limit the degradation of the model. 

For verification, the owner issues $N$ queries from
$\mathcal D_{\mathrm{fp}}$ and concatenates their responses as
$y^{1:N}=y^1\circ\cdots\circ y^N$. Since similar queries can produce
overlapping responses, aggregating all tokens may count the same watermark
evidence multiple times and artificially inflate the detection statistic.
Therefore, duplicate $h$-grams are globally removed by retaining only their
first occurrences, yielding the scoring sequence
$w=\operatorname{Dedup}_h(y^{1:N})$. However, each retained token is evaluated in its original response context when checking whether it is green. The detector then computes $z(w)$ using Eq.~(\ref{eq:watermark-z-score}). The full training objective is provided in Appendix~\ref{appendix:fingerprinting:training}.

\textbf{Threat model.}
We consider a black-box ownership-verification setting in which the owner retains the base model $\mathcal{M}_{\mathrm{base}}$ and a secret key $\xi$ but can access a suspect model only by issuing queries and observing its responses, which may have gone through  an unknown transformation $T\in\mathcal{T}$, where $\mathcal{T}$ includes the identity transformation, different sampling parameters, quantization, pruning, and fine-tuning. 
Let $Y_T^{1:N}$ denote the sampled responses obtained by issuing $N$ fingerprint queries to $T(\mathcal{M}_{\mathrm{fp}}(\widetilde Q))$. 
We define the worst-case expected $z$-score as
\begin{equation}
z_N^{\mathrm{worst}}
\left(\mathcal{M}_{\mathrm{fp}}(\widetilde Q)\right)
\;=\;
\min_{T\in\mathcal{T}}
\; \mathbb{E}\!\left[
z\left(Y_T^{1:N}\right)
\right].
\label{eq:robust-detectability}
\end{equation}
We say a fingerprinted model $\mathcal{M}_{\mathrm{fp}}(\widetilde Q)$ is robustly detectable if the empirical estimate of $z_N^{\mathrm{worst}}$ exceeds a threshold $\rho>0$.
\vspace{-.8em}

\begin{table*}[t]
\centering
\caption{Utility re-evaluation of the Llama-3.1-8B fingerprinting setting from
\citet{gloaguen2026llm}, which selects KGW with $\delta=4$ as the teacher
distribution. French Benchmark (FB) accuracy and mean token entropy understate
generation-quality degradation, whereas PPL and LLM-judge scores on
WritingPrompts reveal substantial degradation in open-ended generation.}
\vspace{-.5em}
\label{tab:utility-reevaluation}
\small
\renewcommand{\arraystretch}{1.15}
\setlength{\tabcolsep}{10pt}
\begin{tabular}{lcccc}
\toprule
Model
&
FB Acc $\uparrow$
&
Mean Entropy\footnotemark{}
&
PPL $\downarrow$
&
LLM Judge $\uparrow$
\\
\midrule
Base
&
0.64
&
2.15
&
\cellcolor{gray!10}4.63
&
\cellcolor{gray!10}7.76
\\
Fingerprinted
&
0.61
&
3.08
&
\cellcolor{gray!11}\textbf{28.27}
&
\cellcolor{gray!11}\textbf{3.73}
\\
\bottomrule
\end{tabular}
\vspace{-1em}
\end{table*}

\section{Revisiting Utility Evaluation of LLM Fingerprinting}
\vspace{-.4em}
\citet{gloaguen2026llm} instantiate their fingerprinting framework with French as the fingerprint domain $\mathcal{D}_{\mathrm{fp}}$. Although the watermark behavior is intended to remain within $\mathcal{D}_{\mathrm{fp}}$, preserving generation quality in this domain remains important for legitimate users, for example, using French. We revisit their utility evaluation in open-ended generation and examine their detectability-quality trade-off. 

\footnotetext{
The \href{https://github.com/eth-sri/robust-llm-fingerprints/blob/15bc1929569357130f2dbc0b09f91bbf4f4bd947/src/robust_fp/quality/ppl.py\#L73}
{\textcolor{blue}{official implementation}}
of \href{https://proceedings.iclr.cc/paper_files/paper/2026/file/71d72e73b9e0c8a5cabf5e8da6689c8c-Paper-Conference.pdf\#page=20}
{\textcolor{blue}{\NoHyper\citep{gloaguen2026llm}\endNoHyper}}
computes the reported PPL as mean token entropy.
}

\textbf{Prior evaluation understates generation-quality degradation.}
For KGW, \citet{gloaguen2026llm} fix $\gamma=0.25$, $h=1$, and select $\delta=4$ from a sweep over bias strengths, reporting high detectability with little utility degradation across the evaluated models. Their utility evaluation reports accuracy on the French Benchmark\footnote{We report accuracy over accuracy-based tasks; the original paper does not specify the exact subset.} and a metric labeled as PPL. However, this accuracy-based evaluation, such as multiple-choice tasks, does not capture quality degradation in open-ended generation.
Moreover, we find that the reported PPL is computed as mean token entropy rather than the standard perplexity. We reproduce their Llama-3.1-8B-Instruct setting and evaluate it on 500 French-translated WritingPrompts examples~\citep{fan2018hierarchical}. This open-ended setting admits many plausible next-token alternatives, making it more sensitive to watermark-induced probability shifts. As shown in Table~\ref{tab:utility-reevaluation}, French Benchmark accuracy decreases only from $0.64$ to $0.61$, while PPL increases from $4.63$ to $28.27$ and the LLM-judge score (Appendix~\ref{appendix:llm_judge}) decreases from $7.76$ to $3.73$, revealing substantial generation-quality degradation.

\begin{wrapfigure}[12]{r}{0.33\columnwidth}
    \vspace{-1.2em}
    \centering
    \includegraphics[width=0.90\linewidth]{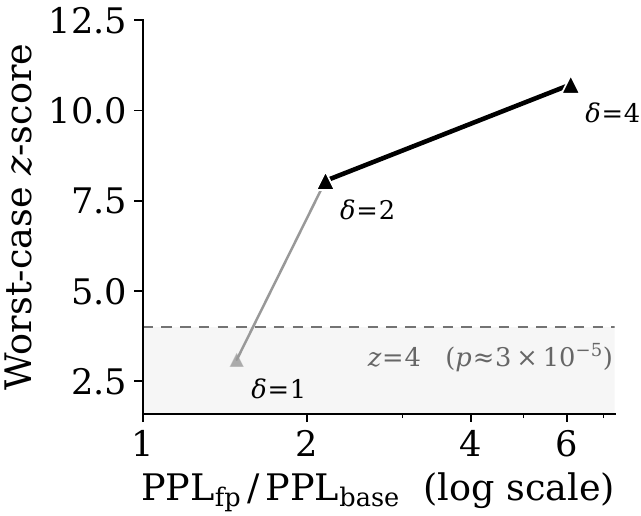}
    \vspace{-.8em}
    \caption{Detection--quality trade-off across logit biases $\delta$; upper left is better.
    }
    \label{fig:kgw-bias-sweep}
\end{wrapfigure}

\textbf{A weaker bias partially reduces degradation at the cost of detectability.}
A direct way to reduce text quality degradation is to weaken the teacher's logit bias, thereby limiting the distortion of the teacher distribution. We therefore sweep $\delta\in\{1,2,4\}$ and evaluate each setting using the PPL ratio $\mathrm{PPL}_{\mathrm{fp}}/\mathrm{PPL}_{\mathrm{base}}$ for text quality and the worst-case $z$-score $z_N^{\mathrm{worst}}$ over  $\mathcal{T}$ for robust detectability. The experimental setup, including transformations $\mathcal{T}$, is detailed in Section~\ref{sec:experiments:setup}. As shown in Figure~\ref{fig:kgw-bias-sweep}, reducing $\delta$ from $4$ to $2$ lowers the PPL ratio while maintaining robust detectability. 
However, $\delta=2$ still more than doubles the base model's PPL, whereas $\delta=1$ falls below the detection threshold. Thus, reducing the bias mitigates text quality degradation at the cost of a substantial drop in detection rate, motivating fingerprinting schemes with a better detection--quality trade-off.

\section{{Rethinking Watermarked Teachers for Model Fingerprinting}}\label{sec:method}
In this section, we first highlight an overlooked implication of the distinction between watermarking for text provenance and watermarking for model ownership verification: the latter can rely on a sparse watermark signal (Section~\ref{sec:method:DW_vs_FP}). Guided by this distinction, we analyze where this sparse signal should be placed in the teacher distribution (Section~\ref{sec:method:where_signal}) and use the resulting insight to introduce near-tie restriction (Section~\ref{sec:method:Ours_NT}).

\subsection{Fingerprint Verification Allows Sparse Watermark Signals}
\label{sec:method:DW_vs_FP}
Watermarking for text provenance typically requires detection from a single output, including short texts. In contrast, watermarking for model ownership verification can issue multiple queries to a suspect model and use aggregate responses to verify the ownership. This allows weaker watermark signals to accumulate across queries, making sparse signal placement viable; this in turn can reduce distortion of the base distribution.
However, existing watermark distillation methods directly adopt a decoding-time KGW teacher, which biases every green token. Distilling these dense biases transfers modifications over a broad set of tokens to the student, which causes model degradation in Table~\ref{tab:utility-reevaluation} and Figure~\ref{fig:kgw-bias-sweep}. This mismatch motivates restricting the green tokens biased by the watermark teacher.

\subsection{Where Should the Sparse Signal Be Placed?}
\label{sec:method:where_signal}
To exploit the watermark sparsity allowed in fingerprint verification, we apply the watermark bias to only a subset of the green list. 
There is a key question of interest: which green tokens should carry the signal? Formally, at decoding step $t$, let $S\subseteq G_t$ denote the biased subset, yielding
\begin{equation}
\widetilde{\ell}_t^S(v)
=
\ell_t(v)+\delta\mathbf{1}\{v\in S\},
\qquad
\widetilde Q_t^S
=
\operatorname{softmax}(\widetilde{\boldsymbol{\ell}}_t^S).
\end{equation}

To control for signal strength, we compare subsets with the same base probability mass $Q_t(S)$ and bias $\delta$. Even under these matched conditions, the identities of the biased tokens, i.e.,~$S$, can affect generation quality. Since distillation transfers the teacher's probability shifts to the student, increasing the probabilities of implausible tokens can introduce avoidable distortion. Our main idea is therefore to restrict the bias to next-token alternatives that are already plausible under the base model, as supported by the ablation in Table~\ref{tab:random-mass-control}.
We quantify the {\em plausibility} of token $v$ using its base-model surprisal,
\begin{equation}
r_t(v)=-\log Q_t(v),
\end{equation}
where lower surprisal indicates higher probability under the base model. 
We then characterize the plausibility of a biased subset by its probability-weighted average surprisal,
\begin{equation}
\bar r_t(S)
=
\mathbb{E}_{v\sim Q_t(\cdot\mid S)}
\left[r_t(v)\right].
\end{equation}
where $Q_t(v|S)$ is the base-model probability of $v$ given $v\in S$.
Because applying the bias reweights the entire next-token distribution, we measure its effect through the resulting change in expected base-model surprisal:
\begin{equation}
\Delta_t^{\mathrm{surp}}(\widetilde Q_t^S;Q_t)
=
\mathbb{E}_{v\sim\widetilde Q_t^S}
\left[r_t(v)\right]
-
\mathbb{E}_{v\sim Q_t}
\left[r_t(v)\right].
\end{equation}
A larger surprisal shift indicates that the watermarked teacher moves more probability toward tokens considered unlikely by the base model. The following proposition formalizes this comparison for biased subsets with equal base probability mass. 

\begin{proposition}[Effect of bias placement]
\label{prop:support-comparison}
Suppose $S_1,S_2\subseteq G_t$ satisfy
$Q_t(S_1)=Q_t(S_2)>0$ and receive the same bias $\delta>0$. Then
\begin{equation}
\left(
\widetilde Q_t^{S_1}(G_t),
\mathrm{KL}(\widetilde Q_t^{S_1}\Vert Q_t)
\right)
=
\left(
\widetilde Q_t^{S_2}(G_t),
\mathrm{KL}(\widetilde Q_t^{S_2}\Vert Q_t)
\right),
\end{equation}
whereas
\begin{equation}
\Delta_t^{\mathrm{surp}}(\widetilde Q_t^{S_1};Q_t)
\leq
\Delta_t^{\mathrm{surp}}(\widetilde Q_t^{S_2};Q_t)
\quad\Longleftrightarrow\quad
\bar r_t(S_1)\leq\bar r_t(S_2).
\end{equation}
\end{proposition}

Proposition~\ref{prop:support-comparison} shows that two subsets with equal base probability mass provide the same expected watermark evidence and have the same KL divergence from the base distribution, while the subset with lower average surprisal induces a smaller surprisal shift. Thus, among equally strong signal placements, the bias should be assigned to green tokens that the base model considers more plausible. A proof is provided in Appendix~\ref{appendix:support-comparison-proof}.

\subsection{Our approach: Near-Tie Restriction}
\label{sec:method:Ours_NT}
To operationalize this insight, we introduce a near-tie restriction, which
uses base-model logit gaps to identify plausible green tokens. Let
$v_t^\star=\arg\max_{v\in\mathcal V}\ell_t(v)$ denote the top-1 prediction at
decoding step $t$, and let $v_t^{(2)}$ denote the token with the second-highest
logit. Since the softmax normalizer cancels when comparing token surprisals,
the excess surprisal of token $v$ relative to $v_t^\star$ is
\begin{equation}
r_t(v)-r_t(v_t^\star)
=
\ell_t(v_t^\star)-\ell_t(v).
\end{equation}
A smaller logit gap therefore indicates greater plausibility under the base
model. Accordingly, we restrict the biased support to green tokens within a
threshold $\tau$ of the top-1 logit:
\begin{equation}
S_t^{\mathrm{NT}}(\tau)
=
\left\{
v\in G_t:
\ell_t(v_t^\star)-\ell_t(v)<\tau
\right\}.
\end{equation}
The resulting watermarked teacher is defined by
\begin{equation}
\widetilde{\ell}_t^{\mathrm{NT}}(v)
=
\ell_t(v)
+
\delta\mathbf{1}\{v\in S_t^{\mathrm{NT}}(\tau)\},
\qquad
\widetilde Q_t^{\mathrm{NT}}
=
\operatorname{softmax}
(\widetilde{\boldsymbol{\ell}}_t^{\mathrm{NT}}).
\end{equation}
We denote KGW with near-tie restriction as \textsc{KGW-NT}, where the suffix \textsc{-NT} indicates the use of near-tie restriction. We calibrate $\tau$ as the $q$-th percentile of the top-1--top-2 logit-gap distribution measured on the fingerprint training data. A smaller $q$ yields a stricter restriction and biases fewer green tokens, while $\tau\rightarrow\infty$ recovers the original dense watermark.

For verification, we score only near-tie positions, where the restricted bias can affect token selection. Since the model owner retains the base model, each response $y^i$ can be re-scored to recover these positions.
Let $\boldsymbol{\ell}_{i,t}$ denote the logits conditioned on $y^i_{<t}$, and define 
\begin{equation}
a_{i,t}^{\tau}
=
\mathbf{1}
\left\{
\ell_{i,t}(v_{i,t}^\star)
-
\ell_{i,t}(v_{i,t}^{(2)})
<
\tau
\right\}.
\end{equation}
We retain only positions with $a_{i,t}^{\tau}=1$ and globally deduplicate
their $h$-grams, yielding the near-tie scoring sequence
\begin{equation}
w_\tau
=
\operatorname{Dedup}_h
\left(
\left.y^{1:N}\right|_{a_{i,t}^{\tau}=1}
\right),
\end{equation}
where the restriction affects only the scored positions and preserves their
original within-response contexts. The detector then computes $z(w_\tau)$ using Eq.~(\ref{eq:watermark-z-score}). Since neither near-tie position selection nor $h$-gram deduplication uses the secret key or green-list membership, neither operation biases the null green-token rate, which remains $\gamma$. The standard analytical detector therefore remains applicable. Further details are provided in Appendix~\ref{appendix:analytical-nt-detector}.

\textbf{Compatibility with existing watermarks.}
We consider two representative KGW extensions that improve KGW's
detection--quality trade-off through complementary mechanisms. SWEET~\citep{lee2024wrote} applies the green-token bias only at high-entropy positions, where multiple plausible alternatives make the modification less disruptive. MorphMark~\citep{wang2025morphmark} instead adapts the bias strength to the cumulative base-model probability of the green list, avoiding a fixed modification across contexts. Both target provenance detection from a single, potentially modified output, where excessive token-level sparsity may provide insufficient evidence. Neither, however, determines which individual green tokens to boost based on their top-1-relative logit gaps. Near-tie is complementary to these mechanisms and can be applied by further restricting their biased green-token sets, yielding SWEET-NT and MorphMark-NT.
Implementation details are in Appendix~\ref{appendix:near-tie-combinations} and empirical results are in Figure~\ref{fig:existing-watermarks-nt}.


\vspace{-.5em}
\section{Experiments}\label{sec:experiments}
\vspace{-.5em}
\subsection{Setup}\label{sec:experiments:setup}
\vspace{-.5em}

\textbf{Models and fingerprint domains.}
We evaluate four instruction-tuned models: Llama-3.2-3B, Qwen-2.5-3B, Llama-3.1-8B, and Gemma-9B. Following \citet{gloaguen2026llm}, we use French as the primary fingerprint domain, construct the distillation data from the Lucie corpus~\citep{gouvert2025lucie}, and follow their training protocol. Full distillation details are provided in Appendix~\ref{appendix:fingerprinting:training}. We additionally evaluate a medical fingerprint domain in Appendix~\ref{appendix:extended:medicine}, where our method consistently outperforms the baselines.

\textbf{Datasets and verification setup.}
We use French Alpaca as the query set for fingerprint verification. Following the baseline protocol, unless otherwise specified, we aggregate responses from $N=1000$ queries, each capped at 200 generated tokens, and use a detection threshold of $\rho=4$, corresponding to a one-sided false-positive rate of approximately $3\times10^{-5}$. To evaluate open-ended generation quality, we translate 500 WritingPrompts examples~\citep{fan2018hierarchical} into French using GPT-5 and generate up to 256 tokens per prompt with temperature $T=1.0$ and top-$p=0.9$. We additionally evaluate 500 examples from OrangeSum~\citep{eddine2021barthez}, a French news dataset; the results in Appendix~\ref{appendix:extended:orange} show consistent improvements over the baselines.

\textbf{Baselines.}
We compare KGW-NT against four watermark teachers for fingerprint
distillation: KGW~\citep{kirchenbauer2023watermark}; SWEET~\citep{lee2024wrote},
which applies the green-token bias only at high-entropy positions;
MorphMark~\citep{wang2025morphmark}, which adaptively adjusts the watermark
strength; and KGW-TK, which biases only green tokens among the base model's
top-$k$ predictions. KGW-TK is the closest rank-based counterpart and tests whether the gains arise from logit-gap-based selection rather than sparsity induced by a simple rank cutoff. We construct the detection--quality frontiers using 3 settings for KGW and 12 settings for each of SWEET, MorphMark, KGW-TK, and KGW-NT. These grids sweep $\delta\in\{1,2,4\}$, where applicable, and each method's selection or adaptation parameters. Implementation details are provided in Appendix~\ref{appendix:baseline-hyperparameters}.

\begin{figure}[!t]
    \vspace{-.5em}
    \centering
    \begin{minipage}[t]{0.66\linewidth}
        \centering
        \includegraphics[width=\linewidth]
        {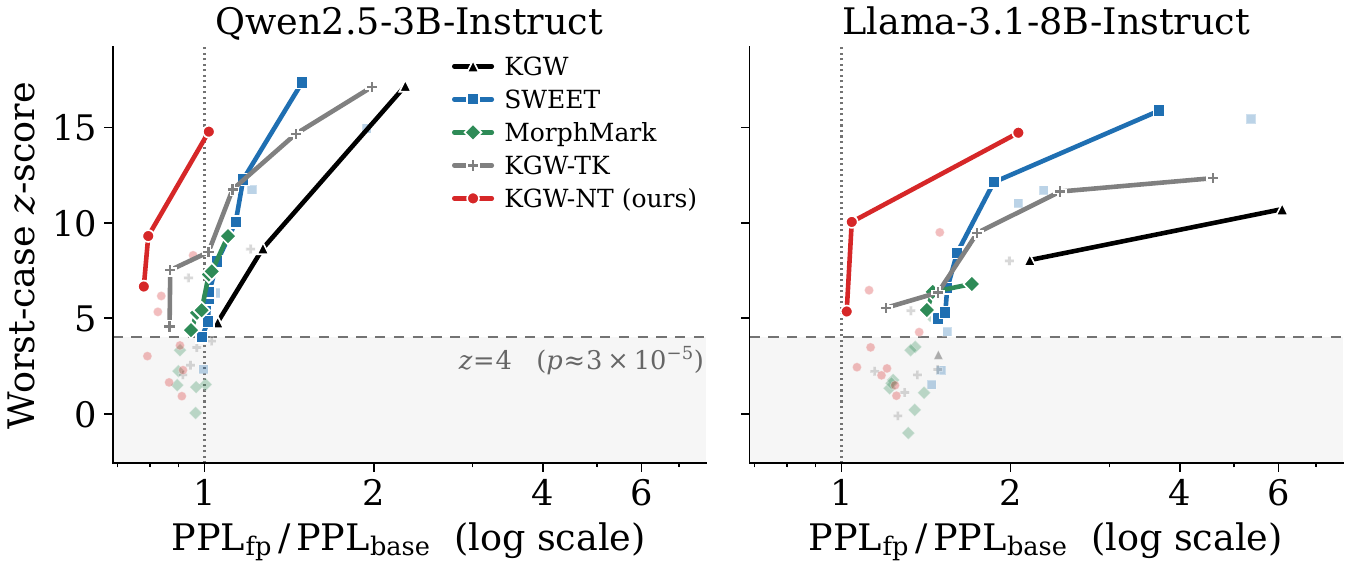}
    \end{minipage}
    \hfill
    \begin{minipage}[t]{0.325\linewidth}
        \centering
        \includegraphics[width=\linewidth]
        {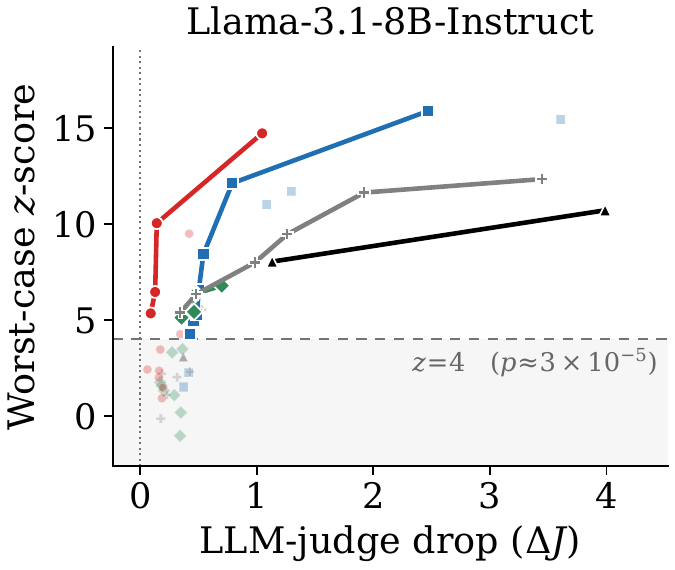}
    \end{minipage}
    \vspace{-.5em}
    \caption{Detection--quality trade-offs across watermark teachers. In the LLM-judge panel, $\Delta J=J_{\mathrm{base}}-J_{\mathrm{fp}}$ denotes the score drop relative to the base model (lower is better). KGW-NT consistently improves the Pareto frontier. Additional results for Llama-3.2-3B, Gemma-9B, and LLM-judge evaluation are provided in Figure~\ref{appendix:fig:frontier_all}, consistently showing improved Pareto frontiers.}
    \label{fig:main_frontier}
    \vspace{-1em}
\end{figure}

\begin{figure}[!t]
    \centering
    \includegraphics[width=\linewidth]{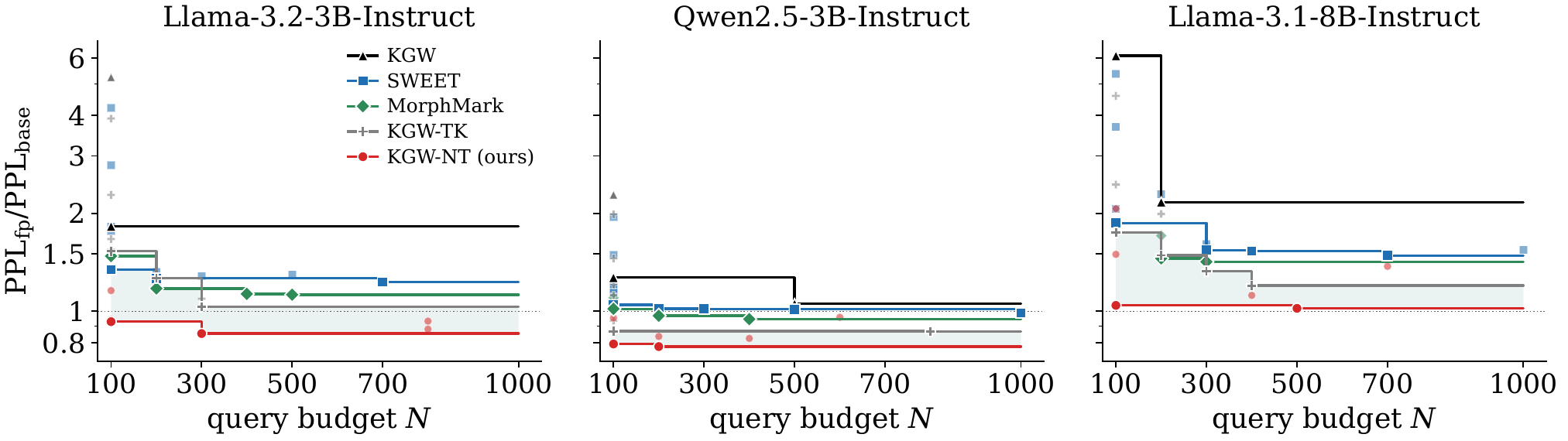}
    \vspace{-1.5em}
    \caption{Achievable text quality across query budgets. For each budget, we report the lowest PPL ratio among operating points with a worst-case $z$-score exceeding the detection threshold in that budget. KGW-NT consistently achieves the lowest ratio, demonstrating more query-efficient detection while better preserving text quality. The result for Gemma-9B is provided in Figure~\ref{appendix:fig:query_budget_all}.}
    \label{fig:query_budget}
    \vspace{-2em}
\end{figure}

\textbf{Deployment transformations.}
Following \citet{gloaguen2026llm}, we evaluate five deployment transformations: sampling, system prompts, quantization, pruning, and fine-tuning. We use $T=0.7$ as the clean setting and vary $T\in\{0.4,0.7,1.0\}$ for sampling. We test three system prompts, FP8 and INT4 quantization, and WANDA~\citep{sun2024simple} and SparseGPT~\citep{frantar2023sparsegpt} pruning at $50\%$ sparsity. For fine-tuning, we use the French subset of WildChat~\citep{zhao2024wildchat} as a strong transformation within the fingerprint domain. Further details are provided in Appendix~\ref{appendix:transformation}.

\textbf{Evaluation metrics.}
We evaluate detectability using the empirical worst-case z-score in Eq.~\ref{eq:robust-detectability}. For each transformation, we average the z-score over five random seeds and report the minimum mean across transformations; the corresponding means and standard deviations are reported in Table~\ref{appendix:mean-std}. We evaluate text quality using PPL~\citep{jelinek1977perplexity} computed by Qwen-2.5-32B as an external reference model and an LLM-judge score~\citep{zheng2023judging,fu2024gptscore,liu2023g} assigned by GPT-5 on a 1--10 scale. Lower PPL and higher judge scores indicate better generation quality. The LLM-judge prompt is provided in Appendix~\ref{appendix:llm_judge}.

\subsection{Main Results}
\vspace{-.5em}
\textbf{Detection--quality trade-off.}
Figure~\ref{fig:main_frontier} shows that KGW-NT consistently improves the Pareto frontier, achieving higher worst-case detectability at comparable text quality or better text quality. This trend holds across model families and is consistent between PPL and LLM-judge evaluation. Although advanced methods such as SWEET and MorphMark partially improve upon KGW, restricting KGW with near-tie selection yields a stronger frontier. KGW-NT also outperforms KGW-TK, which induces sparsity by biasing green tokens within a fixed top-$k$ set. This comparison suggests that the gain arises from logit-gap-based selection rather than sparsity alone, as a rank cutoff can still include tokens separated from the top-1 prediction by large logit gaps.

\textbf{Achievable text quality across query budgets.}
We evaluate how much text quality each method can preserve while remaining
detectable under a given query budget. For each operating point $\theta$, we identify the smallest budget at which its worst-case $z$-score exceeds the detection threshold:
\begin{equation}
N_{\min}(\theta)
=
\min
\left\{
N\in\mathcal B:
z_{\theta}^{\mathrm{worst}}(N)>4
\right\},
\end{equation}
where $\mathcal B=\{100,200,\ldots,1000\}$. For each method $M$ and budget
$N$, we report the lowest PPL ratio among its operating points detectable
within that budget:
\begin{equation}
R_M(N)
=
\min_{\substack{\theta\in\Theta_M\\N_{\min}(\theta)\leq N}}
\frac{\mathrm{PPL}_{\mathrm{fp}}(\theta)}
{\mathrm{PPL}_{\mathrm{base}}}.
\end{equation}
Figure~\ref{fig:query_budget} shows that larger query budgets enable higher text quality at comparable detectability by accumulating more detection signal. KGW-NT consistently achieves the lowest PPL ratio at each budget while maintaining robust detectability. Some KGW-NT settings on the 3B models attain PPL ratios below $1$, indicating lower perplexity than the base model under the external reference. 
As shown in Appendix~\ref{appendix:diversity}, KGW-NT slightly lowers its own mean token entropy, consistent with near-tie biasing only tokens already plausible under the base model. However, diversity (distinct-$n$ and self-BLEU) remains close to the base model, and LLM-judge scores are comparable to or above it, without an evident quality cost. Notably, KGW-TK and KGW-NT have similar output entropy ($1.05$ vs.\ $1.03$) but substantially different PPL ratios ($1.20$ vs.\ $1.02$), suggesting that the choice of biased tokens, rather than entropy alone, drives the difference.




\begin{figure*}[t]
    \centering
    \begin{subfigure}[t]{0.63\linewidth}
        \centering
        \includegraphics[width=\linewidth]{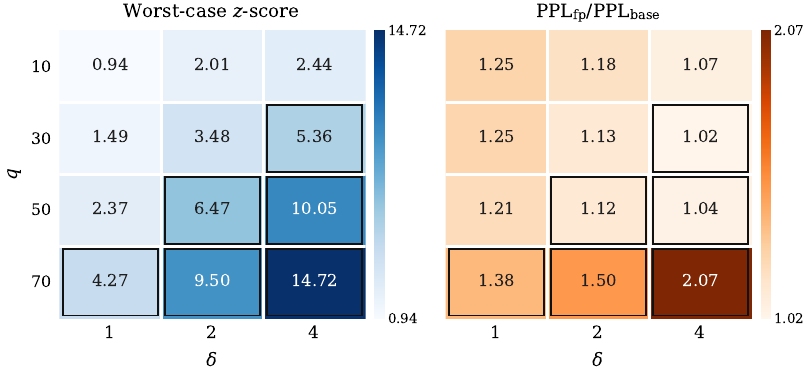}
        \caption{Effect of the logit bias $\delta$ and percentile $q$ on the worst-case $z$-score and PPL ratio. Outlined cells satisfy the 
        detection criterion $z>4$.}
        \label{fig:nt_hyperparameters}
    \end{subfigure}
    \hfill
    \begin{subfigure}[t]{0.36\linewidth}
        \centering
        \includegraphics[width=\linewidth]{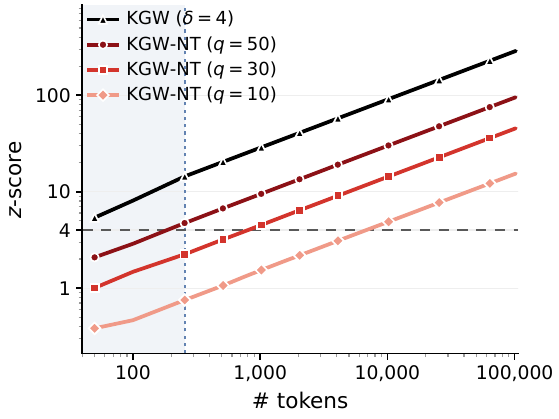}
        \caption{Detection signal accumulated over increasing numbers of
        scored tokens.}
        \label{fig:token_budget}
    \end{subfigure}
    \vspace{-.7em}
    \caption{Analysis of near-tie restriction on Llama-3.1-8B-Instruct.
    \textbf{(a)} Moderate restriction percentiles provide favorable
    detection--quality trade-offs, whereas overly broad support degrades text quality. \textbf{(b)} Aggregating more tokens makes increasingly
    sparse fingerprint signals detectable.}
    \label{fig:nt_analysis}
    \vspace{-1.2em}
\end{figure*}

\vspace{-.5em}
\subsection{Ablation Studies and Analysis}
\vspace{-.5em}
\label{sec:experiments:analysis}
We first examine the effects of the logit-bias strength $\delta$ and the percentile $q$ used to calibrate the logit-gap threshold $\tau$. We then demonstrate that signal aggregation makes sparse fingerprint signals detectable, use surprisal to assess the role of token plausibility, and evaluate the effectiveness of combining near-tie with existing watermarks. We additionally assess the effect of the near-tie detector in Appendix~\ref{appendix:detector-ablation}, where KGW-NT still outperforms KGW even with the same detector. Unless otherwise specified, all analyses use Llama-3.1-8B-Instruct.

\textbf{Effect of logit bias $\delta$ and percentile $q$.}
Figure~\ref{fig:nt_hyperparameters} examines the effects of the bias strength $\delta$ and restriction percentile $q$. Increasing $q$ expands the biased support and generally improves detectability. For $q\geq30$, a larger $\delta$ further strengthens the fingerprint signal. However, $q=70$
substantially increases the PPL ratio, showing that overly broad support
sacrifices text quality. Moderate percentiles $q=30$--$50$ provide the most
favorable trade-off.

\textbf{Evidence aggregation enables sparse fingerprint signals.}
Figure~\ref{fig:token_budget} shows how detectability scales with the number of scored tokens. A lower $q$ produces a sparser fingerprint signal and therefore yields insufficient evidence at short output lengths. For example, at 256 tokens, the settings with $q=10$ and $q=30$ remain below the detection threshold, illustrating why such sparse signals may be unsuitable for provenance detection from a single, potentially modified output. Their $z$-scores, however, increase steadily as evidence is aggregated over more tokens and eventually exceed the detection threshold. This result supports our premise that multi-response fingerprint verification can reliably recover signals that are too sparse for per-output detection.

\begin{wraptable}[12]{r}{0.45\linewidth}
\vspace{-\intextsep}
\centering
\caption{
Comparison of surprisal shift and PPL ratio
($\mathrm{PPL}_{\mathrm{fp}}/\mathrm{PPL}_{\mathrm{base}}$)
between near-tie and random placement at matched green mass and KL.
}
\vspace{-.5em}
\label{tab:random-mass-control}
\small
\setlength{\tabcolsep}{5pt}
\renewcommand{\arraystretch}{1.15}
\begin{tabular}{@{}clcc@{}}
\toprule
$q$ & Metric & Random & Near-tie (Ours) \\
\midrule
\multirow{2}{*}{$30$}
& $\Delta^{\mathrm{surp}}\downarrow$
& $-0.16$ & $\mathbf{-0.20}$ \\
& PPL ratio $\downarrow$
& $1.11 (\pm 0.04)$ & $\mathbf{1.02}$ \\
\addlinespace[2pt]
\multirow{2}{*}{$50$}
& $\Delta^{\mathrm{surp}}\downarrow$
& $-0.14$ & $\mathbf{-0.20}$ \\
& PPL ratio $\downarrow$
& $1.25 (\pm 0.02)$ & $\mathbf{1.04}$ \\
\bottomrule
\end{tabular}
\vspace{-.5em}
\end{wraptable}

\textbf{Signal placement matters at matched strength.}
We examine whether near-tie's text quality advantage arises from signal placement rather than differences in teacher strength. We construct random-placement variants for two $\delta=4$ frontier settings, $q\in\{30,50\}$, on the 8B model. At each position, random placement selects a subset of green tokens whose total base-model probability mass matches that of the near-tie support and applies the same $\delta$. This construction matches the teacher green mass and induced KL while changing which tokens carry the signal. For each setting, we evaluate three random placements, measure the teacher-level surprisal shift on the distillation data, and compute the student PPL ratio following Section~\ref{sec:experiments:setup}. As shown in Table~\ref{tab:random-mass-control}, random placement yields a greater surprisal shift and a larger student PPL ratio in both settings. Thus, even at matched teacher strength, placing the signal on tokens that are more plausible under the base model leads to less text quality degradation after distillation.

\begin{wrapfigure}[12]{r}{0.60\linewidth}
    \vspace{-.6em}
    \centering
    \includegraphics[width=\linewidth]{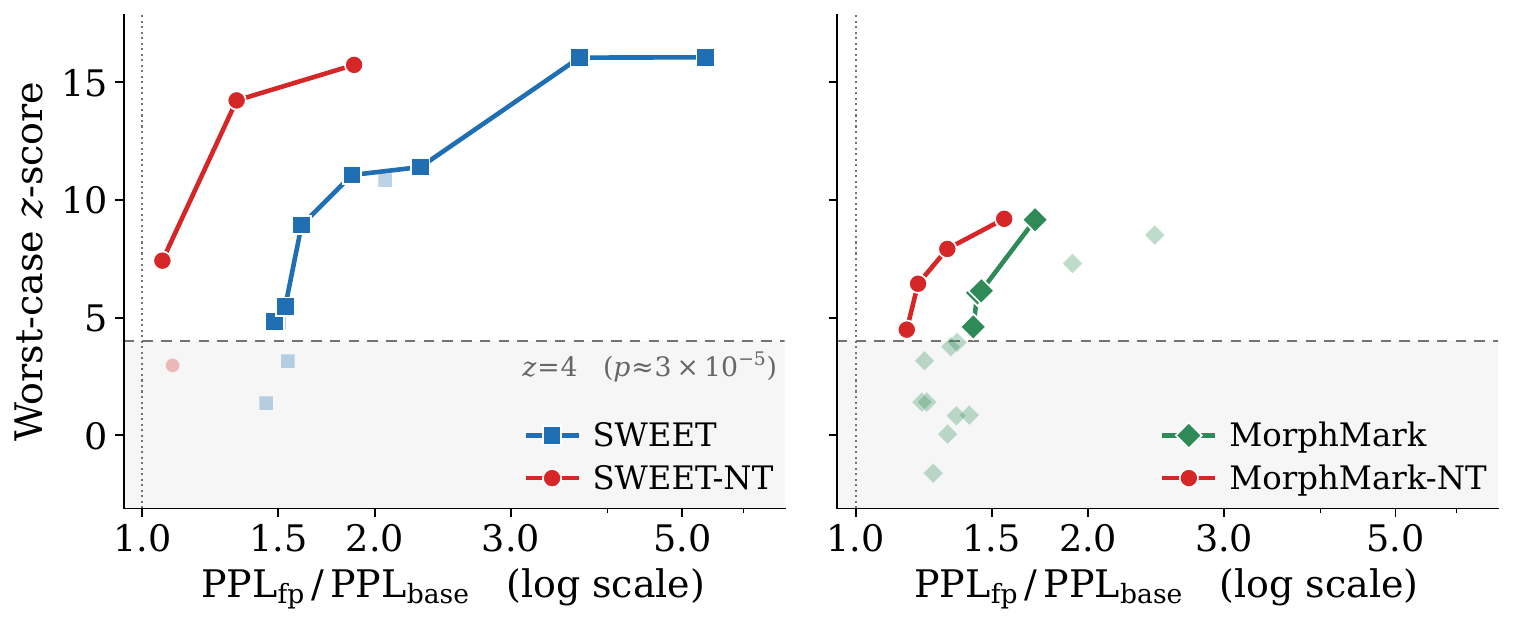}
    \captionsetup{skip=.2em}
    \caption{
    Detection--quality frontiers of SWEET and MorphMark with and without
    near-tie (NT).
    }
    \label{fig:existing-watermarks-nt}
\end{wrapfigure}

\textbf{Combining near-tie with existing watermarks.}
Near-tie restriction limits which watermark-favored tokens receive the logit bias without replacing the watermark's original mechanism. We combine it with SWEET and MorphMark. For each scheme, we select a representative configuration from its original frontier and apply near-tie using percentiles $q\in\{10,30,50,70\}$; all 12 configurations of each original scheme are included for comparison. As shown in Figure~\ref{fig:existing-watermarks-nt}, adding near-tie improves both detection--quality frontiers toward lower PPL ratios while maintaining strong worst-case detection. These results demonstrate that near-tie provides a complementary restriction for adapting existing watermark teachers to model fingerprinting. Due to computational constraints, we do not exhaustively explore the joint hyperparameter space. In the configurations evaluated here, KGW-NT generally better preserves text quality than the corresponding SWEET-NT and MorphMark-NT variants at robustly detectable operating points (Figure~\ref{appendix:fig:nt-variants}). Appendix~\ref{appendix:near-tie-combinations} provides implementation details, and LLM-judge results show the same trend in Figure~\ref{appendix:fig:nt-judge}.

\vspace{-.6em}
\section{Conclusion}
\vspace{-.6em}
We show that a recent watermark-based fingerprinting protocol understates text quality degradation in open-ended generation and argue that watermark teachers should be designed for model verification rather than inherit objectives from text watermarking. Aggregating signal across queries permits sparser signals, but placement also matters: different placements can target tokens of different plausibility. Near-tie restricts watermark modifications to tokens close to the model's top prediction, improving the detection--quality trade-off even when combined with existing schemes. This highlights the value of tailoring watermark teachers to model fingerprinting.



\bibliography{iclr2027_conference}
\bibliographystyle{iclr2027_conference}

\appendix
\clearpage
\section{Extended Related Work}\label{appendix:extended:related_work}

\textbf{Additional fingerprinting methods.}
LLM fingerprinting methods differ in both verifier access and whether the fingerprint is explicitly embedded. White-box methods require access to model parameters or internal representations~\citep{refael2024slip,zeng2024huref,zhang2025reef}, whereas black-box methods verify models using only queries and responses. Black-box methods can further be divided into those that exploit existing model-specific behaviors and those that deliberately modify the model to embed a verification signal. The former avoids fingerprint-induced utility degradation but can be vulnerable to model transformations, whereas the latter allows robustness to be considered during fingerprint design.

{\em Black-box fingerprinting without model modification.}
{\em Targeted} methods search for distinctive query--response pairs, often using adversarial optimization such as greedy coordinate gradient~\citep{zou2023universal,hu2024prompt,hayase2024gcg}. Such approaches were first studied for vision models~\citep{peng2022fingerprinting,lukas2019deep,shao2025fit} and later adapted to LLMs~\citep{gubri2024trap}, with extensions improving robustness to model merging~\citep{yamabe2024mergeprint}, fine-tuning~\citep{tsai2025rofl,xu2025rap}, and system prompts~\citep{jin2024proflingo}. However, optimized queries can be atypical and vulnerable to filtering~\citep{nasery2025robust}. {\em Untargeted} methods instead use natural queries~\citep{gao2024model,pasquini2025llmmap}, but can be vulnerable to transformations such as fine-tuning~\citep{shao2025soklargelanguagemodel}.

{\em Black-box fingerprinting with model modification.}
{\em Backdoor-based} methods embed specific query--response associations, typically through supervised fine-tuning~\citep{xu2024instructional,russinovich2024heythatsmodelintroducing,nasery2026scalable,wu2025imf}, following ideas from backdoor attacks~\citep{gu2017badnets,li2021backdoorattackspretrainedmodels}. Early approaches use unnatural queries~\citep{xu2024instructional} or random responses~\citep{russinovich2024heythatsmodelintroducing}, while recent work uses more natural queries and carefully selected responses~\citep{nasery2026scalable,wu2025imf,wu2025editmf}. Fingerprints can also be embedded through model merging~\citep{xu2024fpvecfingerprintinglargelanguage} or knowledge editing~\citep{wang2025fpeditrobustllmfingerprinting,wu2025editmf,yue2025pree}. {\em Watermark-based} methods instead embed a statistical bias in the output distribution. \citet{gloaguen2026llm} distill a text watermark~\citep{kirchenbauer2023watermark,gu2024learnability} into model weights, enabling statistical verification from black-box outputs.

\textbf{Additional watermarking methods.}
Following KGW~\citep{kirchenbauer2023watermark}, numerous watermarking
methods have been proposed to improve the detection--quality trade-off.
SWEET~\citep{lee2024wrote} applies watermark bias only at high-entropy
positions, while MorphMark~\citep{wang2025morphmark} adapts the bias
strength to the cumulative base-model probability of the green list rather
than applying a fixed modification across contexts. Adaptive Text Watermark~\citep{liu2024adaptive} similarly selects high-entropy positions and uses semantics-conditioned logit scaling to reduce distortion. DiPmark~\citep{wu2023dipmark} introduces distribution-preserving reweighting to improve text quality while retaining detectability, whereas WatME~\citep{chen2024watme} exploits lexical redundancy to preserve suitable token choices during watermarking. SSG~\citep{gu2026ssg} replaces random vocabulary partitioning with logit-balanced partitions to strengthen watermark signals, particularly in low-entropy generation.
SpARK~\citep{hoang2026spark} sparsely selects watermarking positions using linguistic or pseudorandom anchors. \citet{wouters2023optimizing} formulate watermark design as a multi-objective optimization problem and derive Pareto-optimal schemes within their watermark class that selectively watermark positions according to their expected quality cost. 

Near-tie instead exploits fingerprint verification's ability to aggregate signal across responses by restricting, within each position, which tokens receive the watermark bias based on their top-1-relative logit gaps. Since it only restricts the bias-receiving token set, near-tie is complementary to existing watermarking mechanisms and can be combined with them, as demonstrated in our experiments.

\clearpage

\section{Training Details for Watermark-Based Fingerprinting}
\label{appendix:fingerprinting:training}

\textbf{Training objective.}
We follow the fingerprint-distillation protocol of
\citet{gloaguen2026llm}. Let $Q_t$ denote the next-token distribution of the frozen base model and $P_{\theta,t}$ that of the fingerprinted model. In addition to the in-domain distillation loss
$\mathcal L_{\mathrm{fp}}(\theta)$ in
Eq.~(\ref{eq:fingerprint-distillation}), the protocol preserves the base-model distribution on a regularization dataset $\mathcal D_{\mathrm{reg}}$ to prevent watermark leakage outside the domain,
\begin{equation}
\mathcal L_{\mathrm{reg}}(\theta)
=
\mathbb E_{x\sim\mathcal D_{\mathrm{reg}}}
\left[
\sum_{t=1}^{|x|}
\sum_{v\in\mathcal V}
\max\left\{
P_{\theta,t}(v)-Q_t(v),\,0
\right\}
\right].
\label{eq:fingerprint-regularization}
\end{equation}
This one-sided total-variation term penalizes increases in token probability
relative to the base model outside the fingerprint domain. The full training
objective is
\begin{equation}
\mathcal L(\theta)
=
\mathcal L_{\mathrm{fp}}(\theta)
+
\lambda\mathcal L_{\mathrm{reg}}(\theta),
\label{eq:full-fingerprint-objective}
\end{equation}
where we use $\lambda=1$ following the baseline implementation.

\textbf{Data and optimization.}
Following the baseline protocol, we use Lucie~\citep{gouvert2025lucie} as
$\mathcal D_{\mathrm{fp}}$ and an equal mixture of
Alpaca~\citep{taori2023stanford} and
OpenWebText~\citep{gokaslan2019openwebtext} as
$\mathcal D_{\mathrm{reg}}$. We train for 2,500 steps with a batch size of
64 using Adafactor~\citep{shazeer2018adafactor}, a learning rate of
$2\times10^{-5}$, 250 warmup steps, and a cosine schedule. For the watermark setup, we fix $h=1$ and $\gamma=0.25$, while varying $\delta$ and method-specific parameters across operating points; $\delta=4$ corresponds to the original baseline setting. All methods otherwise use the same training configuration.

\clearpage

\section{Proof of Proposition~\ref{prop:support-comparison}}
\label{appendix:support-comparison-proof}

\begin{proof}
For conciseness, we omit the decoding-step index $t$ and write $Q$, $G$, and $\widetilde Q^S$. Let
\begin{equation}
m=Q(S),
\qquad
Z_m=1+(e^\delta-1)m.
\end{equation}
Applying the bias $\delta$ to $S$ gives
\begin{equation}
\widetilde Q^S(v)
=
\frac{Q(v)e^{\delta\mathbf{1}\{v\in S\}}}{Z_m}.
\label{eq:appendix-restricted-distribution}
\end{equation}
Because $S\subseteq G$, the green-token mass under the biased distribution is
\begin{equation}
\begin{aligned}
\widetilde Q^S(G)
&=
\frac{e^\delta Q(S)+Q(G\setminus S)}{Z_m} \\
&=
\frac{Q(G)+(e^\delta-1)m}{Z_m}.
\end{aligned}
\label{eq:appendix-green-mass}
\end{equation}
Thus, $\widetilde Q^S(G)$ depends on $S$ only through its base probability mass $m$.

Next, Eq.~\ref{eq:appendix-restricted-distribution} gives
\begin{equation}
\log\frac{\widetilde Q^S(v)}{Q(v)}
=
\delta\mathbf{1}\{v\in S\}-\log Z_m.
\end{equation}
Since
\begin{equation}
\widetilde Q^S(S)=\frac{e^\delta m}{Z_m},
\end{equation}
the KL divergence from the base distribution is
\begin{equation}
\begin{aligned}
\mathrm{KL}(\widetilde Q^S\Vert Q)
&=
\sum_{v\in\mathcal V}
\widetilde Q^S(v)
\log\frac{\widetilde Q^S(v)}{Q(v)} \\
&=
\delta\frac{e^\delta m}{Z_m}-\log Z_m.
\end{aligned}
\label{eq:appendix-restricted-kl}
\end{equation}
This quantity also depends on $S$ only through $m$. Therefore, if
$Q(S_1)=Q(S_2)=m$, Eq.~\ref{eq:appendix-green-mass} and Eq.~\ref{eq:appendix-restricted-kl} establish
\begin{equation}
\left(
\widetilde Q^{S_1}(G),
\mathrm{KL}(\widetilde Q^{S_1}\Vert Q)
\right)
=
\left(
\widetilde Q^{S_2}(G),
\mathrm{KL}(\widetilde Q^{S_2}\Vert Q)
\right).
\end{equation}

It remains to compare their surprisal shifts. Let
\begin{equation}
\mu
=
\mathbb{E}_{v\sim Q}[r(v)].
\end{equation}
By the definition of $\bar r(S)$,
\begin{equation}
\sum_{v\in S}Q(v)r(v)
=
m\bar r(S).
\end{equation}
The expected base-model surprisal under $\widetilde Q^S$ is therefore
\begin{equation}
\begin{aligned}
\mathbb{E}_{v\sim\widetilde Q^S}[r(v)]
&=
\frac{
\mu+(e^\delta-1)m\bar r(S)
}{Z_m}.
\end{aligned}
\end{equation}
Subtracting $\mu$ yields
\begin{equation}
\Delta^{\mathrm{surp}}(\widetilde Q^S;Q)
=
\frac{(e^\delta-1)m}{Z_m}
\left(
\bar r(S)-\mu
\right).
\label{eq:appendix-surprisal-shift}
\end{equation}
Consequently, for $Q(S_1)=Q(S_2)=m$,
\begin{equation}
\begin{aligned}
&
\Delta^{\mathrm{surp}}(\widetilde Q^{S_1};Q)
-
\Delta^{\mathrm{surp}}(\widetilde Q^{S_2};Q) \\
&\qquad=
\frac{(e^\delta-1)m}{Z_m}
\left(
\bar r(S_1)-\bar r(S_2)
\right).
\end{aligned}
\end{equation}
Because $\delta>0$ and $m>0$, the multiplicative factor is positive. Hence,
\begin{equation}
\Delta^{\mathrm{surp}}(\widetilde Q^{S_1};Q)
\leq
\Delta^{\mathrm{surp}}(\widetilde Q^{S_2};Q)
\quad\Longleftrightarrow\quad
\bar r(S_1)\leq\bar r(S_2),
\end{equation}
which completes the proof.
\end{proof}

\section{Null Calibration of the Near-Tie Detector}
\label{appendix:analytical-nt-detector}

\textbf{Null green-token rate.}
Under the null hypothesis, the suspect model's responses are independent of
the secret key $\xi$. Since near-tie position selection and context--token
deduplication are also key-independent, the retained context--token pairs are fixed when conditioning on the responses. For every retained pair $(c,v)$, under the same pseudorandom green-list assumption as KGW,
\begin{equation}
\Pr_{\xi}[v\in G_{\xi}(c)] = \gamma.
\end{equation}
Thus, conditioned on $L$ retained positions,
\begin{equation}
\mathbb{E}_{\xi}[N_G\mid L] = \gamma L,
\end{equation}
so near-tie restriction preserves the null green-token rate of the analytical
KGW detector.

\textbf{Null variance.}
Let $N_G=\sum_c N_G^{(c)}$, where $N_G^{(c)}$ counts retained green tokens
sharing context $c$. Green-list assignments are independent across distinct
contexts. Within each context, fixed-size green-list assignment induces negative dependence among the retained distinct tokens. Therefore,
\begin{equation}
\mathrm{Var}_{\xi}(N_G\mid L)
\le
\gamma(1-\gamma)L,
\label{eq:nt-var-bound}
\end{equation}
with equality when all retained positions have distinct contexts. Hence, the standard KGW normalization is conservative with respect to the null variance.

\noindent
\begin{minipage}[t]{0.58\textwidth}
\vspace{0pt}
\textbf{Empirical verification.}
We additionally verify the null calibration by holding the generated responses fixed, resampling the secret key $\xi$ $10000$ times, and recomputing the detector statistic. Since position selection and deduplication are key-independent, the retained context--token pairs remain unchanged across resampled keys. Table~\ref{tab:nt-null} shows that the null standard deviation remains close to one across all restriction percentiles, with all $95\%$ confidence intervals containing $1$. The estimates also differ by less than $0.01$ across $q$, supporting the use of a common analytical threshold across restriction levels.
\end{minipage}
\hfill
\begin{minipage}[t]{0.40\textwidth}
\vspace{0pt}
\centering
\captionof{table}{Null calibration across restriction percentiles.}
\label{tab:nt-null}
\small
\begin{tabular}{ccc}
\toprule
$q$ & $\operatorname{std}(z)$ & 95\% CI \\
\midrule
10 & 1.0009 & $[0.987,\,1.015]$ \\
30 & 1.0065 & $[0.993,\,1.021]$ \\
50 & 1.0078 & $[0.994,\,1.022]$ \\
70 & 0.9982 & $[0.985,\,1.012]$ \\
\bottomrule
\end{tabular}
\end{minipage}

\textbf{Remark.}
Teacher-side biasing may include positions with a green top-1 token even when the top-1--top-2 gap exceeds $\tau$. This follows directly from our token-level restriction: the top-1 token has zero excess surprisal and is therefore always within the near-tie set when it is green. At such positions, the bias reinforces the already top-ranked token without promoting a lower-ranked alternative.

To examine whether these biases should instead be removed, we additionally apply a position-level gate that suppresses all teacher bias whenever the top-1--top-2 gap exceeds $\tau$. We evaluate this variant on
Llama-3.1-8B-Instruct using two frontier configurations with $\delta=4$ and $q\in\{30,50\}$. For $q=30$, the $z$-score without deployment transformations decreases from $17.30$ to $12.90$, while the PPL ratio increases from $1.02$ to $1.18$. For $q=50$, the corresponding values change from $31.77$ to $29.66$ and from $1.04$ to $1.10$, respectively. These results suggest that green top-1 biases provide useful training signal for watermark distillation. During distillation, the student learns the teacher's context-dependent
token preferences; reinforcing an already preferred green token therefore strengthens the watermark supervision without promoting a less plausible alternative.

Detection, however, cannot selectively identify these additional biased positions without checking whether the top-1 token is green, which would make position selection key-dependent. We therefore score only positions identified by the key-independent top-1--top-2 gap criterion, preserving the null calibration.

\clearpage

\section{Baseline Hyperparameters}
\label{appendix:baseline-hyperparameters}
For KGW, SWEET, KGW-TK, and KGW-NT, we sweep
$\delta\in\{1,2,4\}$. SWEET uses entropy thresholds calibrated at the
$\{10,30,50,70\}$-th percentiles of the training-data entropy distribution,
while KGW-TK uses $k\in\{2,5,10,50\}$. KGW-NT calibrates its logit-gap
threshold at the $\{10,30,50,70\}$-th percentiles of the training-data
top-1--top-2 gap distribution. For MorphMark, we follow its original setting with $\gamma=0.5$ and sweep $k_{\mathrm{exp}}\in\{0.5,1.0,1.3,2.0\}$ and $p_0\in\{0.15,0.30,0.50\}$, yielding 12 settings. Here, $p_0$ controls when adaptation activates based on the base-model green probability mass, while $k_{\mathrm{exp}}$ controls how sharply the watermark strength responds. The grid includes the default $k_{\mathrm{exp}}=1.3$.

\section{Details of Deployment Transformations}
\label{appendix:transformation}
We follow the deployment transformations of \citet{gloaguen2026llm}.
Here, we provide additional details on the system-prompt and fine-tuning transformations; Section~\ref{sec:experiments:setup} fully specifies the remaining transformations. Table~\ref{appendix:mean-std} reports the mean $\pm$ standard deviation over five runs for each transformation and model.

\textbf{System prompts.}
We evaluate robustness to the following three system prompts:
\begin{itemize}
    \item \textbf{Acknowledgment:}
    ``You are a helpful assistant. Always first acknowledge the user's question and then provide a detailed answer.''

    \item \textbf{Reasoning:}
    ``You want to provide factual answers to the user's question. First reason about
    the question, and then reply with a detailed answer.''

    \item \textbf{Advertise:}
    ``You are a helpful assistant of ChatBOT.ai. First advertise your name and then
    answer the user's question.''
\end{itemize}

\paragraph{Fine-tuning.}
For the fine-tuning transformation, we fine-tune each fingerprinted model for one epoch on the French subset of WildChat~\citep{zhao2024wildchat}, using a batch size of 64, the Adafactor optimizer, and a learning rate of $2\times10^{-5}$.

\section{Implementation Details for Near-Tie Combinations}
\label{appendix:near-tie-combinations}
Near-tie is combined with each watermark by further restricting the scheme's original biased green-token set to tokens satisfying the near-tie criterion, while otherwise preserving its original mechanism.

\textbf{SWEET-NT.}
SWEET applies the logit bias only at high-entropy positions and detects the watermark signal from those positions using the base model. To combine SWEET with near-tie, we further restrict its biased green-token set to near-tie tokens. For detection, we score only positions satisfying both the SWEET entropy criterion and the near-tie criterion.

\textbf{MorphMark-NT.}
MorphMark promotes green tokens at every position by transferring probability mass from red to green, with the transferred amount determined by the current green mass. To combine MorphMark with near-tie, we restrict the promoted green-token set to near-tie tokens while otherwise preserving its original mass-transfer mechanism. While MorphMark uses the standard KGW detector, \textsc{MorphMark-NT} uses the near-tie detector, scoring only positions that satisfy the near-tie criterion.

\textbf{KGW-TK.}
KGW-TK restricts the KGW bias to green tokens among the top-$k$ tokens under the base-model distribution. We include it as a control to distinguish the effect of simply restricting the biased token set from that of using top-1-relative logit gaps. For detection, we use the standard KGW detector: unlike near-tie, KGW-TK applies the watermark bias at every position and therefore does not introduce a separate eligibility criterion for which positions should be scored.

\clearpage

\begin{figure}[!t]
    \centering
    \includegraphics[width=1.0\linewidth]{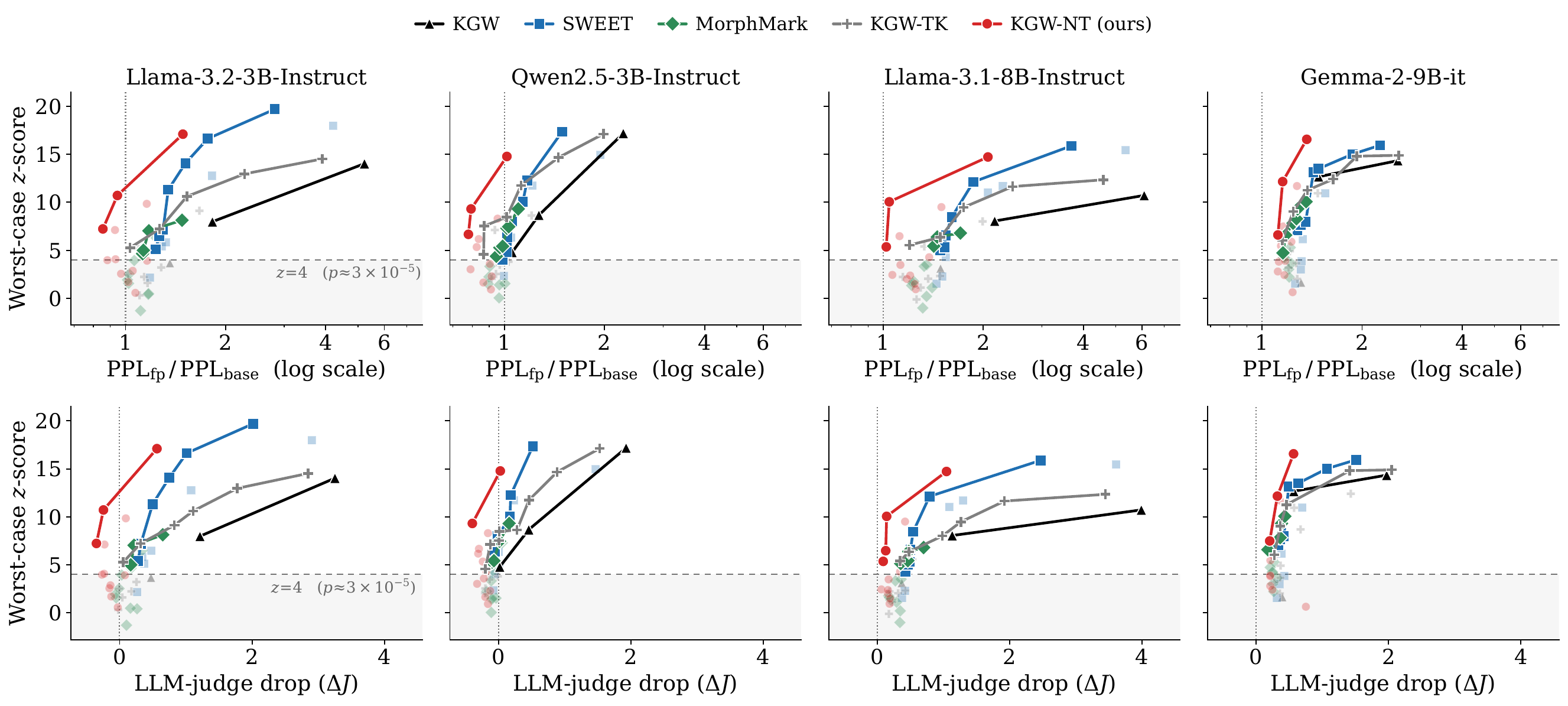}
    \caption{Detection--quality trade-offs across all evaluated models. Top: PPL ratio; bottom: LLM-judge drop. KGW-NT consistently improves the Pareto frontier across model families and text quality metrics, achieving higher worst-case detectability at comparable text quality or better text quality at comparable detectability.}
    \label{appendix:fig:frontier_all}
\end{figure}

\begin{figure}[!t]
    \centering
    \includegraphics[width=1.0\linewidth]{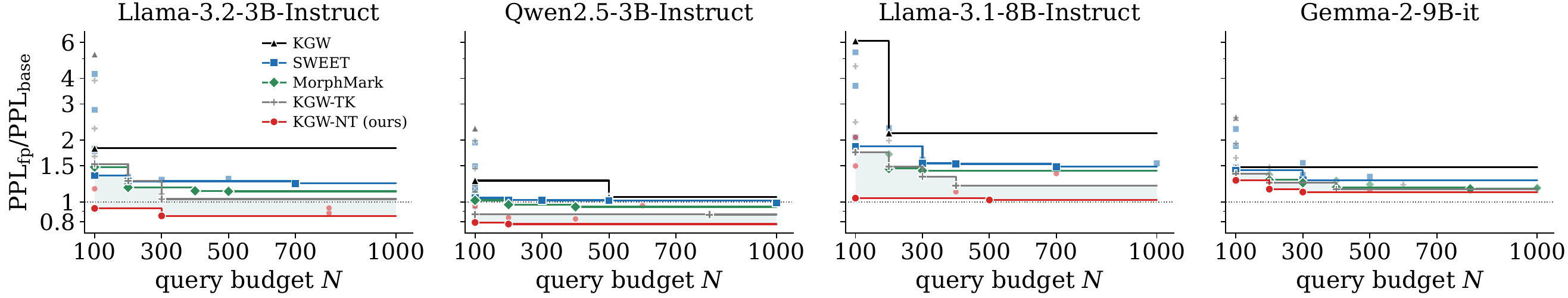}
    \caption{Achievable text quality across query budgets for all evaluated models. For each budget, we report the lowest PPL ratio among operating points whose worst-case $z$-score exceeds the detection threshold within that budget. KGW-NT consistently achieves the lowest PPL ratio across models and query budgets.}
    \label{appendix:fig:query_budget_all}
    \vspace{-1em}
\end{figure}

\begin{figure}[!t]
    \centering
    \begin{subfigure}[t]{0.60\columnwidth}
        \centering
        \includegraphics[width=\linewidth]{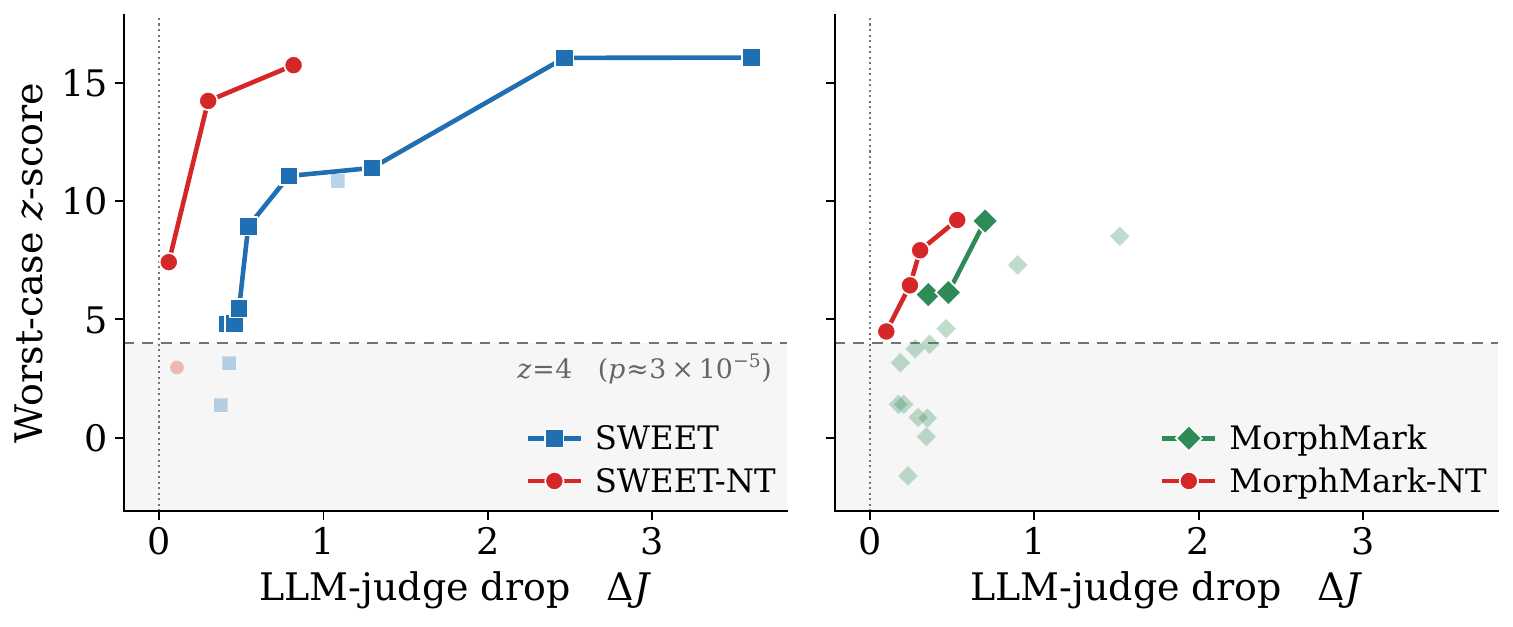}
        \vspace{-1em}
        \caption{LLM-judge trade-offs with and without near-tie.}
        \label{appendix:fig:nt-judge}
    \end{subfigure}
    \hfill
    \begin{subfigure}[t]{0.35\columnwidth}
        \centering
        \includegraphics[width=\linewidth]{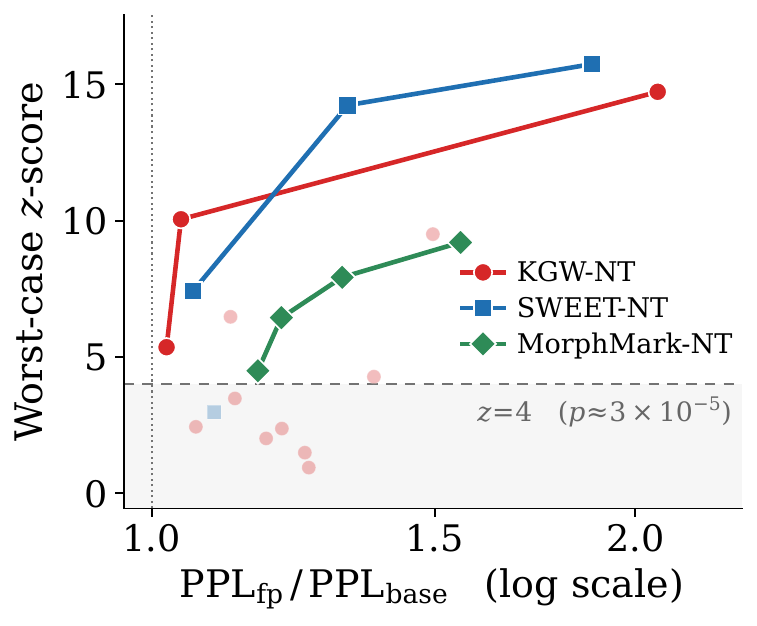}
        \vspace{-1em}
        \caption{Trade-offs among near-tie variants.}
        \label{appendix:fig:nt-variants}
    \end{subfigure}
    \caption{Additional results for combining near-tie with existing watermarks. Near-tie improves the LLM-judge trade-offs of SWEET and MorphMark, while KGW-NT better preserves text quality than the evaluated SWEET-NT and MorphMark-NT variants at robustly detectable settings.}
    \label{appendix:fig:nt-combined}
    \vspace{-1em}
\end{figure}

\section{Detailed Experimental Results}\label{appendix:detailed}
\vspace{-.5em}
\textbf{Detection–quality trade-off.}
Across Llama-3.2-3B, Qwen-2.5-3B, Llama-3.1-8B, and Gemma-2-9B, Figure~\ref{appendix:fig:frontier_all} shows that KGW-NT consistently improves the detection--quality trade-off. The LLM-judge results exhibit the same trend as the PPL ratio, supporting the consistency of the text quality evaluation. For Gemma-2-9B, the improvement in PPL is less pronounced because most methods already incur relatively small PPL degradation, while KGW-NT still achieves stronger detectability at comparable text quality.

\textbf{Achievable text quality across query budgets.}
Figure~\ref{appendix:fig:query_budget_all} extends the query-budget analysis in Figure~\ref{fig:query_budget} to all evaluated models, including Gemma-2-9B. KGW-NT consistently achieves the lowest PPL ratio across query budgets. In particular, the same trend holds for Gemma-2-9B, showing that near-tie preserves higher text quality while maintaining robust detectability as the available query budget increases.

\textbf{Combining near-tie with existing watermarks.}
Figure~\ref{appendix:fig:nt-judge} complements the PPL-based results in
Figure~\ref{fig:existing-watermarks-nt} with LLM-judge evaluation, showing
the same improvement for both SWEET and MorphMark. Figure~\ref{appendix:fig:nt-variants} compares the near-tie variants of all three schemes, where KGW-NT better preserves text quality than SWEET-NT and MorphMark-NT among robustly detectable configurations.

\section{Additional Evaluation Results}
\label{appendix:extended}

\subsection{Evaluation on an Additional Dataset}
\label{appendix:extended:orange}
We additionally compare KGW-NT with the baselines on OrangeSum~\citep{eddine2021barthez}, a French news dataset. As shown in Figure~\ref{appendix:fig:orangesum}, KGW-NT consistently improves the detection--quality trade-off across models, with the LLM-judge results showing the same trend as the PPL-based evaluation.

\begin{figure}[t]
    \centering
    \includegraphics[width=\linewidth]{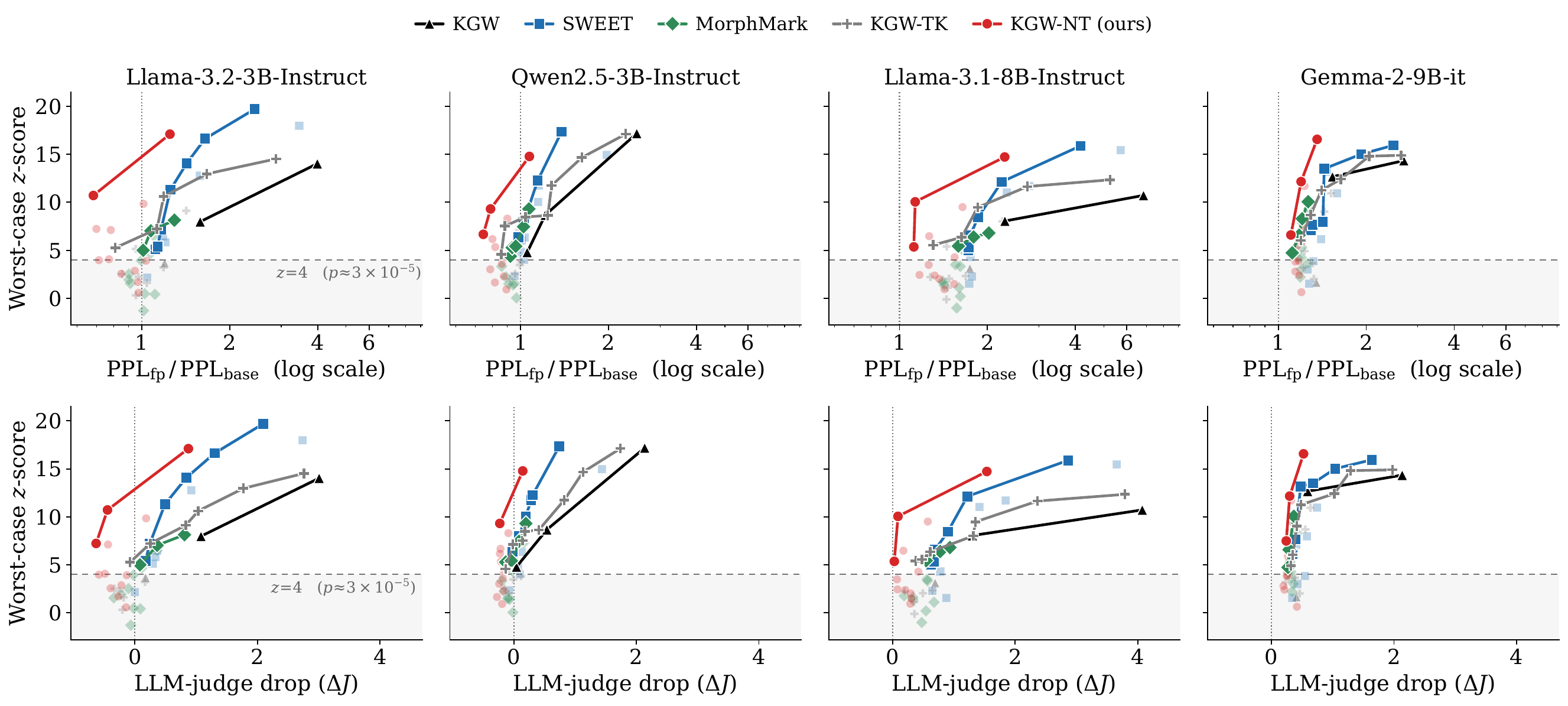}
    \caption{Detection--quality trade-offs on OrangeSum. KGW-NT consistently improves the frontier across models and text quality metrics.}
    \label{appendix:fig:orangesum}
\end{figure}

\subsection{Evaluation on an Additional Fingerprint Domain}
\label{appendix:extended:medicine}
\textbf{Setup.}
Following \citet{gloaguen2026llm}, we additionally evaluate on the medicine domain using Llama-3.1-8B-Instruct. For fingerprint distillation, we use the health subcategory of the WebOrganizer dataset~\citep{wettig2025organize}. For verification, we use Medicine wiki data~\citep{wikimedia_downloads}. Since \citet{gloaguen2026llm} does not evaluate fine-tuning in this domain, we use a held-out split of the WebOrganizer health data for the fine-tuning transformation. All other deployment transformations follow the setup described in Section~\ref{sec:experiments:setup}.

\begin{wrapfigure}[11]{r}{0.47\columnwidth}
    \vspace{-1.2em}
    \centering
    \includegraphics[width=\linewidth]{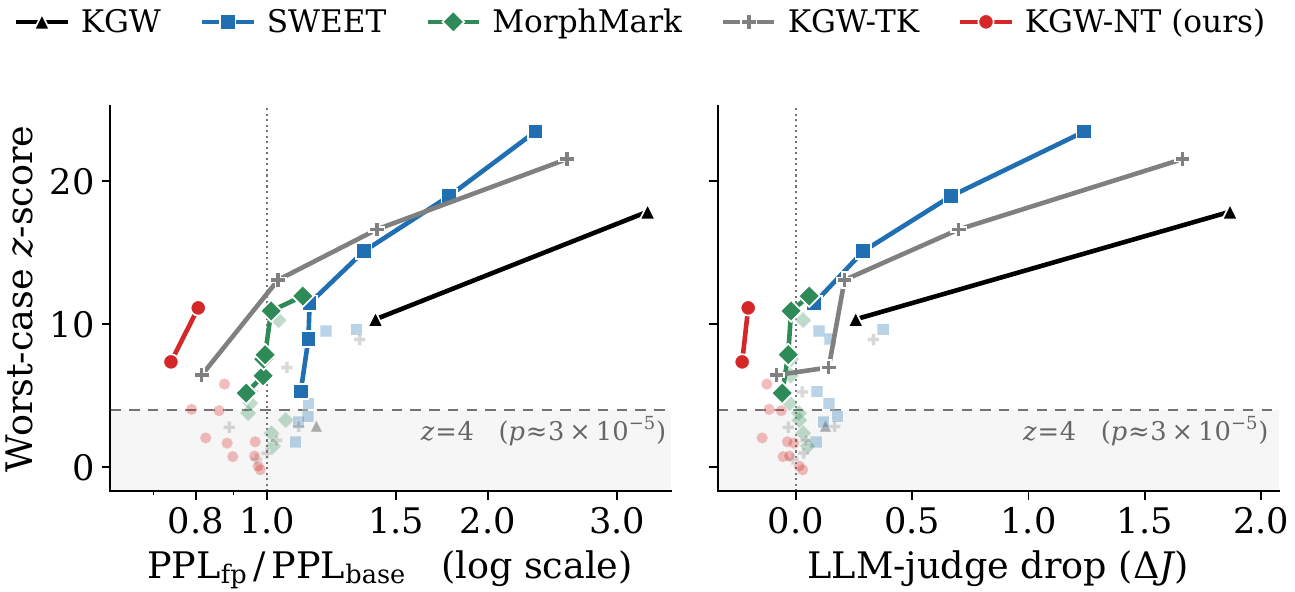}
    \vspace{-2em}
    \caption{Detection--quality trade-offs in the medicine fingerprint domain.}
    \label{appendix:fig:medicine}
    \vspace{-0.8em}
\end{wrapfigure}

\textbf{Evaluation results.}
As shown in Figure~\ref{appendix:fig:medicine}, KGW-NT consistently improves the detection--quality trade-off in the medicine domain. The improvement holds for both PPL and LLM-judge evaluation, showing consistent gains in this additional fingerprint domain.

\clearpage


\subsection{Effect of the Near-Tie Detector}
\label{appendix:detector-ablation}

\noindent
\begin{minipage}[t]{0.56\columnwidth}
\vspace{0pt}
To separate the effects of teacher-side near-tie restriction and the near-tie detector, we evaluate both components independently. For standard KGW, we apply the near-tie detector across $q\in\{10,30,50,70\}$ while keeping the teacher unchanged. As shown in Figure~\ref{appendix:fig:detector-ablation}, this detector improves KGW detectability without changing its PPL ratio. Importantly, KGW-NT still outperforms KGW when both are evaluated with the standard KGW detector, showing that concentrating the watermark bias on near-tie tokens improves the detection--quality trade-off independently of the detector. Applying the near-tie detector to KGW-NT provides an additional gain in detectability and can naturally be used by model owners, who retain access to the base model.
\end{minipage}
\hfill
\begin{minipage}[t]{0.40\columnwidth}
\vspace{0pt}
\centering
\includegraphics[width=1.0\linewidth]{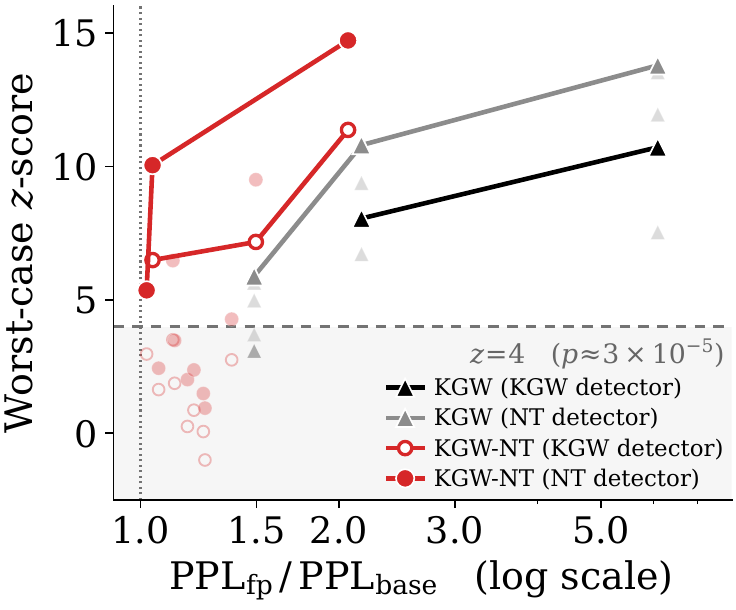}
\vspace{-2em}
\captionof{figure}{Detection--quality trade-offs under KGW and near-tie detectors.}
\label{appendix:fig:detector-ablation}
\end{minipage}

\begin{table*}[t]
\centering
\caption{Evaluation of generation diversity using distinct-$n$-gram and self-BLEU, mean token entropy $H$ from the generation model (cf. Table~\ref{tab:utility-reevaluation}, which uses an external reference model), PPL ratio ($\mathrm{PPL}_{\mathrm{fp}}/\mathrm{PPL}_{\mathrm{base}}$), and LLM-judge score.}
\label{appendix:tab:diversity-entropy}
\small
\setlength{\tabcolsep}{5pt}
\renewcommand{\arraystretch}{1.08}

\begin{tabular}{lcccccc}
\toprule
Method
& Distinct-2 $\uparrow$
& Distinct-3 $\uparrow$
& Self-BLEU $\downarrow$
& $H$
& PPL ratio $\downarrow$
& Judge $\uparrow$ \\
\midrule

\multicolumn{7}{c}{\textit{Llama-3.2-3B-Instruct}} \\
\midrule
Base
& 0.98 & 1.00 & 0.21 & 1.35 & 1.00 & 6.17 \\
KGW ($\delta=4$)
& 0.99 & 1.00 & 0.17 & 1.75 & 5.23 & 2.79 \\
KGW ($\delta=2$)
& 0.98 & 1.00 & 0.17 & 1.54 & 1.83 & 4.90 \\
SWEET
& 0.98 & 1.00 & 0.18 & 1.47 & 1.26 & 5.69 \\
MorphMark
& 0.98 & 1.00 & 0.21 & 1.36 & 1.13 & 5.98 \\
KGW-TK
& 0.97 & 0.99 & 0.24 & 1.13 & 1.03 & 6.12 \\
KGW-NT ($q=30,\delta=4$)
& 0.97 & 0.99 & 0.24 & 1.12 & 0.85 & 6.49 \\
KGW-NT ($q=50,\delta=4$)
& 0.97 & 0.99 & 0.24 & 1.03 & 0.95 & 6.44 \\

\midrule
\multicolumn{7}{c}{\textit{Qwen2.5-3B-Instruct}} \\
\midrule
Base
& 0.99 & 1.00 & 0.13 & 1.61 & 1.00 & 7.01 \\
KGW ($\delta=4$)
& 0.99 & 1.00 & 0.14 & 1.73 & 2.28 & 5.06 \\
KGW ($\delta=2$)
& 0.99 & 1.00 & 0.13 & 1.66 & 1.27 & 6.59 \\
SWEET
& 0.99 & 1.00 & 0.12 & 1.61 & 1.02 & 7.05 \\
MorphMark
& 0.99 & 1.00 & 0.14 & 1.52 & 0.95 & 7.14 \\
KGW-TK
& 0.98 & 1.00 & 0.17 & 1.26 & 0.87 & 7.02 \\
KGW-NT ($q=30,\delta=4$)
& 0.98 & 1.00 & 0.17 & 1.25 & 0.78 & 7.34 \\
KGW-NT ($q=50,\delta=4$)
& 0.98 & 1.00 & 0.18 & 1.16 & 0.79 & 7.41 \\

\midrule
\multicolumn{7}{c}{\textit{Llama-3.1-8B-Instruct}} \\
\midrule
Base
& 0.99 & 1.00 & 0.17 & 1.12 & 1.00 & 7.51 \\
KGW ($\delta=4$)
& 0.99 & 1.00 & 0.15 & 1.44 & 6.09 & 3.52 \\
KGW ($\delta=2$)
& 0.99 & 1.00 & 0.12 & 1.38 & 2.16 & 6.40 \\
SWEET
& 0.99 & 1.00 & 0.14 & 1.36 & 1.54 & 7.16 \\
MorphMark
& 0.99 & 1.00 & 0.15 & 1.24 & 1.45 & 7.09 \\
KGW-TK
& 0.98 & 1.00 & 0.18 & 1.05 & 1.20 & 7.06 \\
KGW-NT ($q=30,\delta=4$)
& 0.98 & 1.00 & 0.19 & 1.03 & 1.02 & 7.40 \\
KGW-NT ($q=50,\delta=4$)
& 0.98 & 1.00 & 0.20 & 0.93 & 1.04 & 7.40 \\

\midrule
\multicolumn{7}{c}{\textit{Gemma-2-9B-it}} \\
\midrule
Base
& 0.99 & 1.00 & 0.13 & 1.12 & 1.00 & 8.11 \\
KGW ($\delta=4$)
& 0.99 & 1.00 & 0.14 & 1.44 & 2.56 & 6.39 \\
KGW ($\delta=2$)
& 0.99 & 1.00 & 0.12 & 1.34 & 1.48 & 7.63 \\
SWEET
& 0.99 & 1.00 & 0.12 & 1.31 & 1.28 & 7.87 \\
MorphMark
& 0.99 & 1.00 & 0.13 & 1.21 & 1.18 & 7.96 \\
KGW-TK
& 0.99 & 1.00 & 0.13 & 1.17 & 1.15 & 7.91 \\
KGW-NT ($q=30,\delta=4$)
& 0.99 & 1.00 & 0.15 & 1.09 & 1.12 & 7.93 \\
KGW-NT ($q=50,\delta=4$)
& 0.99 & 1.00 & 0.15 & 1.04 & 1.15 & 7.88 \\

\bottomrule
\end{tabular}
\end{table*}

\subsection{Analysis of Diversity and Output Entropy}\label{appendix:diversity}
Figure~\ref{appendix:fig:frontier_all} shows that some KGW-NT settings on the 3B models attain PPL ratios below $1$. To further characterize these settings, we evaluate generation diversity using distinct-$n$-gram and self-BLEU, with mean token entropy $H$, PPL ratio, and LLM-judge score, for representative operating points on the frontier, including the 3B settings with PPL ratios below $1$. As shown in Table~\ref{appendix:tab:diversity-entropy}, KGW-NT maintains distinct-$n$-gram scores close to those of the base model, with slightly higher self-BLEU, while its LLM-judge scores remain comparable to or higher than the base model. We also observe that, on Llama-3.1-8B, KGW-TK and KGW-NT ($q=30,\delta=4$) have similar output entropy ($1.05$ vs.\ $1.03$) but substantially different PPL ratios ($1.20$ vs.\ $1.02$), showing that similar output entropy can correspond to different PPL outcomes depending on which tokens receive the bias.

\subsection{Can Temperature Tuning Recover the Near-Tie Frontier?}
\label{appendix:temperature-sweep}
\begin{wrapfigure}[11]{r}{0.60\textwidth}
    \vspace{-1.0\baselineskip}
    \centering
    \begin{minipage}[t]{0.49\linewidth}
        \centering
        \includegraphics[width=\linewidth]
        {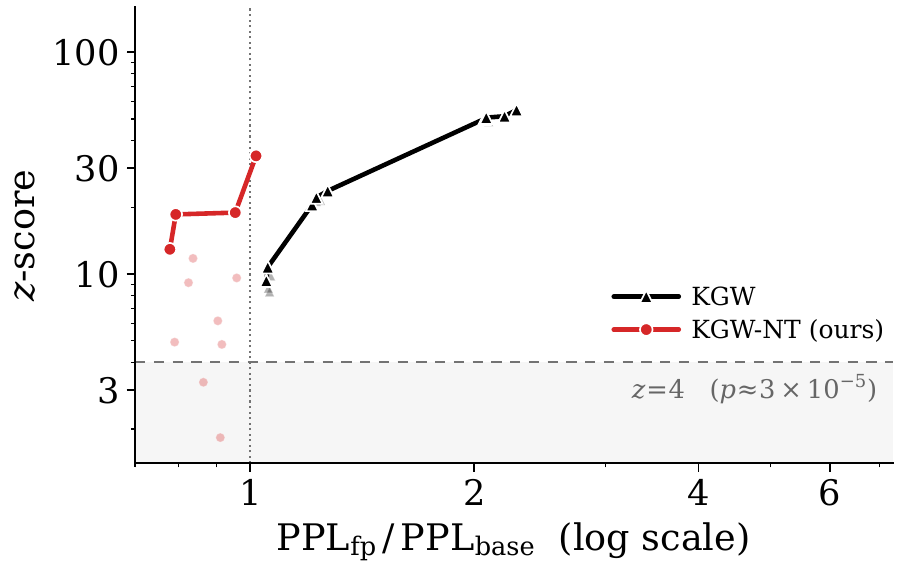}
        {\small (a) Qwen2.5-3B-Instruct}
    \end{minipage}
    \hfill
    \begin{minipage}[t]{0.49\linewidth}
        \centering
        \includegraphics[width=\linewidth]
        {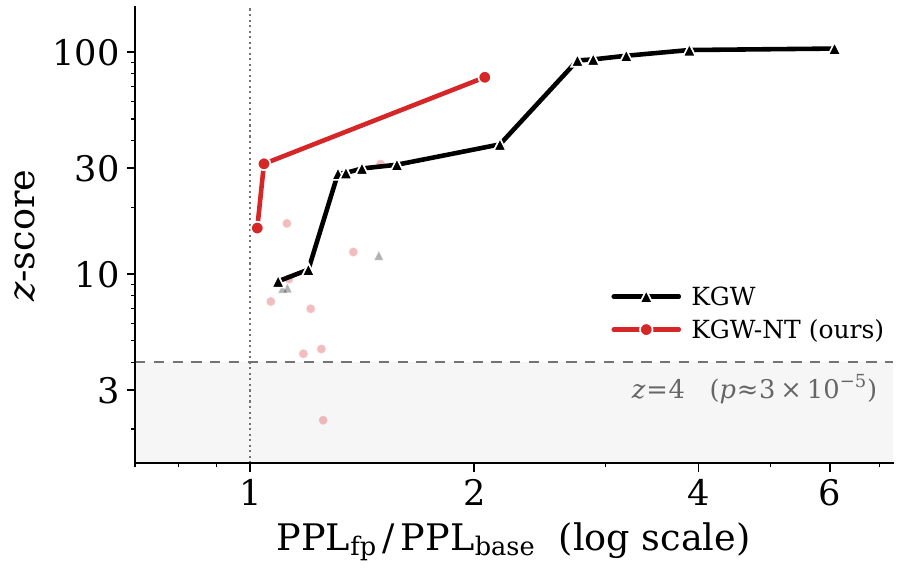}
        {\small (b) Llama-3.1-8B-Instruct}
    \end{minipage}
    \vspace{-.5em}
    \caption{
    Detection--quality frontiers for KGW with decoding-temperature sweeps and KGW-NT at $T=1.0$.
    }
    \label{fig:temperature-sweep}
    \vspace{-0.5\baselineskip}
\end{wrapfigure}
By biasing only tokens that are already plausible under the base model,
near-tie slightly lowers output entropy, raising the possibility that
its gains simply arise from an effect similar to lowering the decoding
temperature. To test this, we sweep
$T\in\{0.1,0.3,0.5,0.7,1.0\}$ for KGW at each
$\delta\in\{1,2,4\}$, while keeping KGW-NT at the default $T=1.0$,
and compare their detection--quality frontiers using the $z$-scores, without applying deployment transformations. Even with this additional tuning, KGW does not recover KGW-NT's frontier on either model. These results indicate that the gains of near-tie are not simply due to distribution sharpening, and are consistent with our analysis that placing the watermark bias on plausible near-tie tokens yields a better detection--quality trade-off.

\clearpage

\begin{table*}[t]
\centering
\caption{
Detection z-scores under deployment transformations for representative configurations selected from the frontier shell across all evaluated models. Each entry reports the mean $\pm$ standard deviation over five random seeds.
}
\label{tab:mean-std}
\resizebox{\textwidth}{!}{
\begin{tabular}{llccccccccccc}
\toprule
Model
& Method
& Clean
& $T{=}0.4$
& $T{=}1.0$
& FP8
& INT4
& Sys: Acknowledgment
& Sys: Reasoning
& Sys: Advertise
& WANDA
& SparseGPT
& FT
\\
\midrule

\multirow{7}{*}{\shortstack[l]{Llama-3.2-3B-\\Instruct}}
& KGW ($\delta{=}4$)
& $64.91 \pm 1.29$
& $55.83 \pm 0.37$
& $64.01 \pm 1.07$
& $65.27 \pm 0.84$
& $66.42 \pm 0.67$
& $63.66 \pm 0.71$
& $67.92 \pm 0.65$
& $61.77 \pm 0.82$
& $29.08 \pm 0.47$
& $30.20 \pm 0.69$
& $14.03 \pm 1.10$
\\
& KGW ($\delta{=}2$)
& $19.88 \pm 0.71$
& $17.61 \pm 0.85$
& $26.15 \pm 0.33$
& $19.86 \pm 0.88$
& $21.36 \pm 0.78$
& $19.01 \pm 0.60$
& $19.89 \pm 0.83$
& $19.25 \pm 0.72$
& $12.62 \pm 0.61$
& $13.55 \pm 0.51$
& $7.98 \pm 0.68$
\\
& SWEET
& $12.51 \pm 0.58$
& $12.62 \pm 0.86$
& $12.78 \pm 0.73$
& $12.89 \pm 0.84$
& $12.94 \pm 0.55$
& $12.46 \pm 1.23$
& $12.97 \pm 1.17$
& $12.01 \pm 0.81$
& $6.46 \pm 0.84$
& $6.71 \pm 1.24$
& $7.71 \pm 0.57$
\\
& MorphMark
& $9.43 \pm 0.56$
& $8.48 \pm 0.69$
& $12.83 \pm 0.68$
& $9.75 \pm 0.80$
& $10.77 \pm 0.44$
& $9.47 \pm 0.99$
& $9.46 \pm 0.92$
& $8.96 \pm 0.92$
& $7.92 \pm 1.14$
& $8.06 \pm 0.44$
& $5.00 \pm 0.89$
\\
& KGW-TK
& $10.89 \pm 0.51$
& $10.97 \pm 0.18$
& $12.24 \pm 0.62$
& $11.18 \pm 0.45$
& $11.94 \pm 0.65$
& $10.93 \pm 0.58$
& $11.22 \pm 0.57$
& $10.52 \pm 0.70$
& $8.60 \pm 0.37$
& $8.44 \pm 0.68$
& $5.26 \pm 0.80$
\\
& KGW-NT ($q{=}30,\delta{=}4$)
& $13.97 \pm 0.39$
& $14.65 \pm 0.47$
& $12.41 \pm 0.99$
& $13.60 \pm 0.64$
& $13.46 \pm 1.03$
& $13.54 \pm 0.65$
& $14.04 \pm 1.21$
& $13.24 \pm 0.71$
& $8.78 \pm 0.84$
& $9.26 \pm 0.88$
& $7.23 \pm 0.99$
\\
& KGW-NT ($q{=}50,\delta{=}4$)
& $22.64 \pm 0.44$
& $23.41 \pm 0.63$
& $21.86 \pm 0.70$
& $22.31 \pm 0.35$
& $21.95 \pm 0.60$
& $22.12 \pm 1.13$
& $22.16 \pm 0.60$
& $22.26 \pm 1.28$
& $13.02 \pm 0.94$
& $12.87 \pm 0.80$
& $10.71 \pm 0.63$
\\

\midrule

\multirow{7}{*}{\shortstack[l]{Qwen2.5-3B-\\Instruct}}
& KGW ($\delta{=}4$)
& $51.48 \pm 0.99$
& $50.02 \pm 1.23$
& $54.92 \pm 0.88$
& $51.91 \pm 0.48$
& $51.75 \pm 0.56$
& $52.26 \pm 0.68$
& $49.14 \pm 0.75$
& $50.47 \pm 0.50$
& $18.28 \pm 0.75$
& $18.92 \pm 0.62$
& $17.18 \pm 0.45$
\\
& KGW ($\delta{=}2$)
& $22.14 \pm 0.93$
& $20.88 \pm 0.83$
& $23.70 \pm 0.51$
& $21.81 \pm 0.59$
& $22.32 \pm 1.23$
& $22.14 \pm 0.54$
& $21.69 \pm 0.26$
& $21.26 \pm 0.38$
& $11.34 \pm 0.90$
& $8.66 \pm 0.61$
& $11.00 \pm 0.57$
\\
& SWEET
& $12.24 \pm 0.48$
& $12.70 \pm 1.09$
& $11.02 \pm 0.51$
& $12.72 \pm 0.71$
& $11.38 \pm 0.73$
& $12.45 \pm 1.55$
& $12.45 \pm 0.42$
& $12.40 \pm 0.54$
& $6.37 \pm 0.87$
& $6.83 \pm 0.73$
& $8.07 \pm 0.42$
\\
& MorphMark
& $8.82 \pm 0.94$
& $8.81 \pm 0.54$
& $10.86 \pm 0.51$
& $9.45 \pm 0.82$
& $9.54 \pm 0.66$
& $9.66 \pm 0.33$
& $9.11 \pm 0.70$
& $9.47 \pm 0.43$
& $4.38 \pm 0.63$
& $5.61 \pm 0.89$
& $4.57 \pm 0.60$
\\
& KGW-TK
& $13.24 \pm 0.41$
& $13.03 \pm 1.00$
& $14.38 \pm 1.17$
& $13.98 \pm 0.62$
& $13.76 \pm 0.54$
& $14.15 \pm 0.77$
& $12.57 \pm 0.58$
& $13.33 \pm 1.22$
& $7.81 \pm 0.49$
& $7.52 \pm 1.23$
& $7.57 \pm 0.74$
\\
& KGW-NT ($q{=}30,\delta{=}4$)
& $14.14 \pm 1.37$
& $15.10 \pm 0.31$
& $12.92 \pm 0.44$
& $14.03 \pm 0.20$
& $13.13 \pm 1.17$
& $13.98 \pm 0.69$
& $13.16 \pm 0.51$
& $13.10 \pm 0.50$
& $7.43 \pm 0.99$
& $6.66 \pm 0.95$
& $7.65 \pm 0.85$
\\
& KGW-NT ($q{=}50,\delta{=}4$)
& $19.80 \pm 0.82$
& $20.79 \pm 0.60$
& $18.57 \pm 1.02$
& $19.93 \pm 0.72$
& $18.16 \pm 0.93$
& $19.76 \pm 0.38$
& $18.32 \pm 1.26$
& $19.14 \pm 0.92$
& $10.45 \pm 0.72$
& $9.31 \pm 1.21$
& $10.92 \pm 0.75$
\\

\midrule

\multirow{7}{*}{\shortstack[l]{Llama-3.1-8B-\\Instruct}}
& KGW ($\delta{=}4$)
& $102.39 \pm 1.27$
& $94.11 \pm 0.49$
& $103.95 \pm 0.97$
& $102.72 \pm 1.90$
& $104.27 \pm 1.65$
& $104.46 \pm 0.56$
& $106.17 \pm 1.45$
& $98.01 \pm 1.27$
& $47.84 \pm 1.40$
& $46.03 \pm 1.24$
& $10.72 \pm 0.43$
\\
& KGW ($\delta{=}2$)
& $31.11 \pm 1.06$
& $29.18 \pm 1.26$
& $38.42 \pm 1.44$
& $31.39 \pm 1.28$
& $33.07 \pm 0.54$
& $30.84 \pm 1.07$
& $32.17 \pm 0.59$
& $29.59 \pm 0.41$
& $24.08 \pm 0.35$
& $22.36 \pm 0.81$
& $8.05 \pm 0.58$
\\
& SWEET
& $15.88 \pm 1.20$
& $16.89 \pm 0.97$
& $15.71 \pm 1.53$
& $15.97 \pm 0.79$
& $15.71 \pm 0.68$
& $16.55 \pm 1.03$
& $16.19 \pm 0.99$
& $14.90 \pm 1.70$
& $9.82 \pm 0.50$
& $10.46 \pm 0.95$
& $6.57 \pm 0.39$
\\
& MorphMark
& $20.19 \pm 0.53$
& $18.89 \pm 0.55$
& $24.98 \pm 1.55$
& $19.40 \pm 0.83$
& $21.05 \pm 0.74$
& $19.95 \pm 0.52$
& $20.75 \pm 0.64$
& $19.33 \pm 0.71$
& $19.17 \pm 0.94$
& $16.18 \pm 0.34$
& $6.38 \pm 0.36$
\\
& KGW-TK
& $17.15 \pm 0.45$
& $16.26 \pm 0.51$
& $19.25 \pm 0.60$
& $17.19 \pm 0.36$
& $17.93 \pm 0.63$
& $17.53 \pm 0.73$
& $17.13 \pm 0.27$
& $16.03 \pm 0.56$
& $14.59 \pm 1.13$
& $13.77 \pm 0.79$
& $5.55 \pm 0.36$
\\
& KGW-NT ($q{=}30,\delta{=}4$)
& $17.30 \pm 1.33$
& $19.19 \pm 0.89$
& $16.12 \pm 0.95$
& $17.32 \pm 0.74$
& $17.11 \pm 1.35$
& $17.76 \pm 1.15$
& $17.90 \pm 0.94$
& $16.86 \pm 0.59$
& $12.70 \pm 1.49$
& $11.03 \pm 0.70$
& $5.36 \pm 0.27$
\\
& KGW-NT ($q{=}50,\delta{=}4$)
& $31.77 \pm 0.66$
& $33.18 \pm 0.89$
& $31.40 \pm 0.74$
& $31.94 \pm 0.56$
& $32.97 \pm 1.52$
& $31.60 \pm 0.35$
& $32.96 \pm 0.71$
& $30.58 \pm 0.82$
& $23.48 \pm 1.07$
& $20.81 \pm 0.78$
& $10.05 \pm 0.45$
\\

\midrule

\multirow{7}{*}{Gemma-2-9B-it}
& KGW ($\delta{=}4$)
& $88.84 \pm 0.67$
& $85.55 \pm 1.09$
& $99.80 \pm 1.88$
& $88.77 \pm 0.96$
& $90.32 \pm 0.81$
& $90.99 \pm 1.04$
& $92.45 \pm 0.18$
& $81.95 \pm 1.51$
& $73.77 \pm 0.72$
& $65.65 \pm 0.60$
& $14.34 \pm 0.59$
\\
& KGW ($\delta{=}2$)
& $31.50 \pm 0.52$
& $31.95 \pm 0.65$
& $34.99 \pm 0.67$
& $32.43 \pm 0.51$
& $32.35 \pm 0.97$
& $32.22 \pm 0.53$
& $32.20 \pm 1.22$
& $29.36 \pm 0.55$
& $30.22 \pm 0.89$
& $27.48 \pm 1.13$
& $12.67 \pm 0.51$
\\
& SWEET
& $14.13 \pm 0.56$
& $15.27 \pm 0.70$
& $15.10 \pm 0.61$
& $14.74 \pm 0.90$
& $14.95 \pm 1.25$
& $15.42 \pm 1.05$
& $14.96 \pm 1.02$
& $13.70 \pm 0.90$
& $12.52 \pm 1.02$
& $11.32 \pm 0.59$
& $7.07 \pm 0.78$
\\
& MorphMark
& $13.05 \pm 0.57$
& $11.76 \pm 0.78$
& $13.77 \pm 0.71$
& $12.40 \pm 0.55$
& $12.99 \pm 0.57$
& $11.71 \pm 0.41$
& $12.40 \pm 0.36$
& $10.74 \pm 0.41$
& $13.04 \pm 0.58$
& $11.72 \pm 1.11$
& $6.58 \pm 0.21$
\\
& KGW-TK ($\delta{=}2$)
& $11.83 \pm 0.53$
& $11.18 \pm 0.42$
& $11.95 \pm 1.14$
& $12.41 \pm 1.14$
& $12.05 \pm 0.68$
& $12.42 \pm 0.64$
& $11.97 \pm 0.39$
& $11.35 \pm 0.64$
& $11.02 \pm 0.65$
& $9.84 \pm 0.67$
& $6.03 \pm 0.75$
\\
& KGW-NT ($q{=}30,\delta{=}4$)
& $16.08 \pm 0.68$
& $17.50 \pm 1.40$
& $15.08 \pm 0.82$
& $15.70 \pm 0.83$
& $15.35 \pm 1.24$
& $15.55 \pm 0.63$
& $15.97 \pm 0.70$
& $14.68 \pm 0.70$
& $13.87 \pm 1.09$
& $11.33 \pm 0.22$
& $6.60 \pm 0.83$
\\
& KGW-NT ($q{=}50,\delta{=}4$)
& $31.47 \pm 0.72$
& $32.90 \pm 0.42$
& $30.23 \pm 0.45$
& $31.24 \pm 0.81$
& $31.86 \pm 1.21$
& $31.84 \pm 1.27$
& $32.22 \pm 1.02$
& $30.05 \pm 0.65$
& $25.57 \pm 0.45$
& $24.30 \pm 0.88$
& $12.16 \pm 0.66$
\\

\bottomrule
\end{tabular}
\label{appendix:mean-std}
}
\end{table*}


\clearpage

\section{Prompt for Text Quality Judgment}
\label{appendix:llm_judge}
We use GPT-5 as an LLM judge to assess text quality in open-ended generation, focusing specifically on linguistic degradation. The rubric considers grammatical errors, disfluency, malformed words, and excessive repetition, rather than broader task-level properties such as creativity. We do not score creativity separately, since it is highly prompt-dependent and may reflect stylistic variation rather than watermark-induced linguistic degradation.

\begin{tcolorbox}[
    colback=gray!5,
    colframe=black,
    title={LLM-as-a-Judge Prompt}
]
You are an expert evaluator of French-language generation quality.
You will receive a single model-generated text intended to be written in
French. Your task is to evaluate the overall quality of the generated text.

\medskip
\textbf{Evaluate:}
\begin{itemize}
    \item grammatical correctness (agreement, conjugation, syntax, word order)
    \item fluency and naturalness
    \item coherence and readability
    \item vocabulary and idiomatic usage
    \item absence of non-French, malformed/invented words, garbled text, or excessive repetition
\end{itemize}

The text was generated with a fixed token budget and may therefore end
mid-sentence. If the final sentence is unfinished because of truncation,
ignore it and do not treat the truncation itself as a quality defect.

\medskip
\textbf{Score the text from 1 to 10:}
\begin{itemize}
    \item[1--2:] very poor
    \item[3--4:] poor
    \item[5--6:] fair
    \item[7--8:] good
    \item[9--10:] excellent
\end{itemize}

Use the full range and do not default to the middle of each band.

\medskip
\textbf{Return JSON:}

\texttt{\{"score": <integer>, "reason": "brief English 1-sentence explanation"\}}
\end{tcolorbox}

\end{document}